\documentclass[aps,longbibliography, 12pt, a4paper,floatfix]{revtex4-1}
\usepackage[utf8]{inputenc}
\usepackage[T1]{fontenc}
\usepackage{graphicx}
\usepackage{amsmath}
\usepackage{amsfonts, amssymb}
\usepackage{amsthm}
\usepackage{times}
\usepackage{dsfont}
\usepackage{placeins}
\usepackage{array,longtable}
\usepackage{booktabs}
\usepackage[table,x11names,dvipsnames]{xcolor}
\usepackage[a4paper, total={6.5in, 9.5in}]{geometry}
\usepackage{hyperref}
\usepackage[nobiblatex]{xurl}
\usepackage{tikz}
\usetikzlibrary {positioning}
\definecolor{processblue}{cmyk}{0.96,0,0,0}
\newtheorem{theorem}{Theorem}[section]

\newtheorem{proposition}[theorem]{Proposition}

\newtheorem{remark}[theorem]{Remark}

\hypersetup{
    colorlinks=true,
    linkcolor=blue,          % color of internal links
    citecolor=blue           % color of citations
}
\def\beq{\begin{equation}}
\def\eeq{\end{equation}}
\def\beqna{\begin{eqnarray}}
\def\eeqna{\end{eqnarray}}
\def\bea{\begin{array}}
\def\ea{\end{array}}

\def\mr{{\mathcal R}}

\newcommand{\etal}{\textit{et al.~}}
\begin{document}
\title{
An epidemiological compartmental model with automated parameter estimation and forecasting of the spread of COVID-19
 with analysis of data from Germany and Brazil
}
\author{Adriano A. Batista $^1$}
\email{adriano@df.ufcg.edu.br}
\author{Severino Horácio da Silva$^2$}
\email{horacio@mat.ufcg.edu.br}
\date{\today}
%~
%\thanks{$^1$Research partially supported by CAPES/CNPq}
\affiliation{$^{1}$ Departamento de Física, Universidade Federal de Campina Grande, 58051-900 Campina Grande PB,
    Brazil.\\
    $^2$ Departamento de Matemática,
    Universidade Federal de Campina Grande, 58051-900 Campina Grande PB,
    Brazil.
    } 
    
\begin{abstract}
In this work, we adapt the epidemiological SIR model to study the evolution of the dissemination of
COVID-19 in Germany and Brazil (nationally, in the State of Paraíba, and in the City of Campina Grande). We prove the well posedness and the continuous dependence of the model dynamics on its  parameters.
We also propose a simple probabilistic method for the evolution
of the active cases that is instrumental for the automatic
estimation of parameters of the epidemiological model.
 We obtained statistical estimates of the active cases
 based the probabilistic method and on the confirmed cases data.
 From this estimated time series we obtained a time-dependent
 contagion rate, which reflects a lower or higher
 adherence to social distancing by the involved populations.
 By also analysing the data on daily deaths, we obtained the daily lethality and recovery rates.
 We then integrate the equations of motion of the model
 using these time-dependent parameters.
 We validate our epidemiological model by fitting the official data of confirmed, recovered, death, and active cases due to
the pandemic with the theoretical predictions.
 We obtained very good fits of the data with this method.
 The automated procedure developed here could be used for
 basically any population with a minimum of extra work.
 Finally, we also propose and validate a forecasting method
 based on Markov chains for the evolution of the epidemiological
 data for up to two weeks.

%We also make epidemiological predictions for the next twenty days.

\end{abstract}
\maketitle
\section{Introduction}
The COVID-19 pandemic already has reached practically the
whole planet.
According to the World Health Organization (WHO, Situation Report-41 \cite{who2020Report41}), although around 80\% of the infected people present mild symptoms
(equivalent to the common flu), older people and those
with a history of other diseases like diabetes, cardiovascular, and chronic respiratory syndrome can develop
serious problems after being infected by the SARS-CoV-2 virus. 
It was first identified in December 2019 in the City of Wuhan, Province of Hubei, China. From there it spread across Asia, Europe, and the other continents \cite{hui2020continuing}.
On March 11th, the WHO declared it to be a pandemic \cite{who2020}.

In this work, we investigate the spread of the pandemic
in 4 different scenarios: Germany, Brazil, the Brazilian
State of Paraíba, and the City of Campina Grande.
We use the data from Germany mostly as a benchmark test for our model, since very likely their COVID-19 data is one of the most accurate in the world, which is due to the widespread testing of its population. Furthermore, the social distancing, isolation, and use of personal protection equipment (PPE) there has been far more efficient than the measures taken in Brazil in containing the spread of COVID-19 infections. 

The first confirmed case in Brazil occurred in February 25th of 2020, when a 61-year-old man who had returned from Italy
tested positive and the first death due to the pandemic occurred on March 12.
It is very unlikely all cases of COVID-19 contaminations evolved
from this patient. 
As reported in other countries, there were multiple imported
infections, which usually drive the contagion rate to very high values in the initial stages of the epidemic.
In Brazil and in the majority of countries around the world
social distancing policies have been adopted in order to
decrease the rate of contagion and thus allow that health
systems do not collapse and have conditions to treat the
gravest cases of the disease.

In the State of Paraíba, according to the Paraíba Department of Health,
the first confirmed case of COVID-19 was registered on March 18th of 2020.
It was a man who lived in the City of João Pessoa that had returned from a trip
to Europe on February 29th. On March 31st the first death due to
the pandemic was recorded in the State.

Campina Grande is the second largest city in the State of Paraíba.
The social distancing policy was implanted in this city on March 20th of 2020,
with the closure of universities, schools, and non-essential stores.
The social distancing was implemented in a preventive form since the first case
of infection by the SARS-CoV-2 virus only came to be registered one week afterwards, on
March 27th. The first death due to COVID-19 in this city was only registered on April 19th.

In this work, we adapted the simple and well-known SIR epidemiological model, developed by
Kermack and Mckendrick in 1927 \cite{kermack1927contribution}, 
to study the evolution of the dissemination of COVID-19.
The SIR model is a well-known and tested epidemiological compartmental model that has been applied
to very diverse epidemics (see for example Refs.  \cite{alcaraz2012modeling, bastos2020modeling, isea2013mathematical, crokidakis2020data}). 
The model receives its name for dividing the population in three groups: susceptible, infected, and removed.
Although, SIR and SIR-like models, such as the SIRD and SEIR(D), are simple
models compared with far more detailed alternatives, such as
proposed by Ndairu \etal  \cite{Ndairou2020},
their simplicity is a strength when it comes to parameter estimation. More complex models have more parameters to be
determined, and hence more uncertainties that may render
them non viable for large scale applications to many populations.
Usually, the available epidemiological data is not complete enough to provide estimates for all the parameters of more complex models.
The model we propose here has three independent parameters that
need to be estimated from the epidemiological data: contagion
rate, average time duration of infection, and lethality rate.
We validated the proposed theoretical model with comparisons of official data of the numbers of confirmed, recovered, death,
and active cases due to COVID-19 infections in Germany, Brazil, the Brazilian State of Paraíba, and the City of Campina Grande, located in Paraíba.
We used a compartmental model called SIRD, which replaces the removed by the recovered and the deceased cases. Furthermore, in the present work we use a time-dependent contagion rate and time-dependent lethality and recovery rates. 

SIR(D) or SEIR(D) models with time-dependent contagion rates are not new \cite{smirnova2019forecasting},
even for the COVID-19 pandemic there are already several SIR and
SIR-like models with time-dependent parameters.
The most difficult part of this approach is the design of a parameter estimation method that is automated and renders the model accurate.
The early work by Fang \etal \cite{Fang2020} used a SEIR model in which they
estimated the parameters based on the epidemiological data using
statistical methods, but at the time of publication there were
only about 40 data points and scant comparison of data with
theoretical model time series.
Zhong \etal \cite{zhong2020early} used a time-dependent
SIR model in which they estimate the parameters from the
data of active and recovered cases.
The authors obtained big fluctuations in the contagion rate
and in the removal rate ($1/\tau$). This likely occurred because
they approximated the time derivative of the infected ($dI/dt$) on a daily basis and also, perhaps, because recovered cases data is often less reliable than the confirmed cases data.
Chen \etal \cite{chen2020time} also used a parameter estimation
technique of the contagion and the removal rates, but they did not provide information on the parameters obtained in their results explicitly.
Dehning \etal \cite{dehning2020inferring} estimated parameters
based on a Bayesian inference with the Markov chain Monte Carlo technique, but not many data points were available at the
time of publication and no estimates for active cases were provided.
%Peixoto \etal \cite{Peixoto2020},
Linka \etal \cite{Linka2020} also used a Bayesian parameter estimate in a SEIR model.
None of the above articles provide information on estimates
for death cases.
A comparison of models (SIR, SEIR, and a branching point process) highlighting the strengths and difficulties of each model, and stressing the importance of nonpharmaceutical public health interventions was made by Bertozzi \etal \cite{Bertozzi2020challenges}.

%made %Fang \etal\cite{Fang2020}.

The present work is new in the way it estimates the parameters
of the epidemiological model dynamical system.
In addition to the adapted SIR model used, we develop a probabilistic model to obtain the recovery and the lethality
probabilities.
Based on this model, we make estimates for the
active cases.
The best fit for the active cases provides us
with the average time duration of infection, which is held fixed in our model.
With the help of this probabilistic model and the
epidemiological data, we also obtain the time-dependent 
parameters of the model: contagion, recovery, and lethality rates.
The estimate of the time-dependent
contagion rate is based on a fairly simple statistical analysis of the statistical estimate of the active cases.
We estimate the time-dependent lethality
and recovery rates based on a statistical analysis of the real death cases data.
In this way, one avoids the time-consuming task, for the programmer, of obtaining the contagion rate function that leads
to the best fit of the data. 
In this work, the best fit of the data is done automatically in one pass of integration.
Thus, this greatly reduces the time taken to fit the available
data with the theoretical curve.
This also allows for short-time forecasting from one to two weeks in advance. 
Furthermore, we present and validate a forecasting method based on Markov chains and on the parameter estimation method we use to make short-term predictions of the evolution of the epidemiological variables.

Additionally, we prove the well posedness of this model (existence, unicity, and continuous dependence on initial data). Also, an important theoretical issue we investigate is the dependence of the solutions on the parameters present in the model. As far as we know, this has not been proved yet for this type of model. Hence, we prove the continuous dependence of the solutions on the system parameters. We use techniques similar to those used in \cite{Silva2020} and \cite{Diekmann2005} to
accomplish this result.

We hope that this study of epidemiological dynamics be useful in stressing the importance of public health policies about the application, maintenance, and reinforcement of social distancing measures with the objective of avoiding the 
collapse of the health system.

We point out that an earlier version of this work was pre-published in ResearchGate \cite{batista2020epidemiological}. 

This work is organized in the following way:
In Section II, we propose our epidemiological model and we prove results on existence and uniqueness of solution and on the continuous dependence of the solution with respect to the initial data and the parameters present in the model. We also
develop the probabilistic model with estimates of the probabilities of
recovery and death. 
We also describe how to make statistical estimates for the various parameters used in our epidemiological model.
In Section III, we investigate the evolution of COVID-19 in Germany as a benchmark test for our model. Afterwards, we present our results and discuss about the evolution of the disease in Brazil, in the State of Paraíba, and in the City of Campina Grande.  
We validate the predictions of our theoretical model by fitting the official data.
% Furthermore, based on the good fits obtained so far, we offer predictions for
% the evolution of the pandemic for the next 20 days in the studied regions. 
Finally, in Section IV we draw our conclusions.

\section{Epidemiological model}
\iffalse
{\color{blue}The original SIR model \cite{kermack1927contribution, murray2007mathematical}
is given by the following system of ordinary differential equations (ODEs)
\beq
\begin{aligned}
 \frac{dS}{dt} & = -\kappa SI,\\
\frac{dI}{dt} & = \kappa SI -\frac I\tau,\\
\frac{dR}{dt} & = \frac I\tau,
\end{aligned}
\label{modeloSIRorig}
\eeq
in which  $S(t)$, $I(t)$, and $R(t)$ represent, respectively, the number of susceptible, infected, removed (recovered or deceased) individuals, 
$\kappa$  represents the rate of contagion and $\tau$ is the average time of infection (from contagion until recovery or death).
}
\fi
The evolution of the epidemiological model we propose is determined
 by the following ordinary differential equations (ODE) system 
\beq
\begin{aligned}
 \frac{dS}{dt} & = \nu(S+I+R)-\mu S-\kappa(t) SI,\\
\frac{dI}{dt} & = -\left(\mu +\frac1\tau\right) I+\kappa(t) SI,\\
\frac{dR}{dt} & = \rho(t) I-\mu R,\\
\frac{dM}{dt} & = \lambda(t) I,
\end{aligned}
\label{modeloSIR}
\eeq
where $S(t)$, $I(t)$, $R(t)$ are, respectively, the normalized
susceptible, infected, and recovered populations at time $t$.
We define $M(t)$ as the normalized number
of accumulated deaths due to the epidemic.
The introduction of $M(t)$ is mostly for convenience, so that we 
do not have to perform a separate integration from the numerical
routine used to integrate the differential equations of our model.
The normalization of the variables was achieved by dividing them by $P_0$, the total initial population of the region considered.
For the sake of simplification, we assume that the population is homogeneous such that all the susceptible individuals have
the same probability of being contaminated and  the infected
individuals have the same probability of recovery or death
due to the infection.
We also suppose that the population evolves in such a way that the newborn babies are all susceptible and the recovered are all immune.
The parameter $\nu$ is the population birth rate, $\mu$ is the
death rate before the onset of the pandemic, $\kappa(t)$ is the contagion rate function, $\rho(t)$ is the recovery rate, and $\lambda(t)$ is the lethality rate due to the epidemic. 
Although $\rho(t)$ and $\lambda(t)$ are changing over time, $\tau^{-1} = \rho(t) + \lambda(t) $ is held constant.

It is important to point out further differences between the original model (\ref{modeloSIR}) and the proposed SIR model \cite{kermack1927contribution}. 
 This could become relevant if the pandemic lasts for over a year.
This is relevant also as a source of comparison with the
death rates due to the COVID-19.
One important difference we introduced is that we now allow for time variation
in the contagion rate $\kappa(t)$, which reflects changes in confinement, social distancing, and mask use adopted by the population.
In addition, we allow for time variation in the recovered and lethality rates, which might reflect possible improvements in the treatment efficacy of COVID-19 and/or changes in the demographics of the infected. Here, we also allow variations in the total population, taking into account the contributions of birth and death rates to the population evolution.

%Before we proceed, we divide all variables involved,
%$S(t)$, $I(t)$, $R(t)$, and $M(t)$ by $P_0$, the total initial population of the region considered, which can be a city, a state, a country or the whole world.
%Hence, $\kappa(t)$ has to be multiplied by $P_0$.
%In this way, if all the other parameters of the model are the same, $\kappa(t)$ becomes independent of $P_0$.

%For a given time $T>0$ and non-negative constants  $\kappa_1<\kappa_2$, $\rho_1 < \rho_2$ and $\lambda_1<\lambda_2$ by Picard-Lindelöf theorem \cite{Hale}, for  each initial value $(S_0,I_0,R_0,M_0)$, with  $\kappa(t)$, $\rho(t)$ and $\lambda(t)$ varying continuously in the bounded intervals $[\kappa_1, \kappa_2]$,  $[\rho_1, \rho_2]$, $[\lambda_1,\lambda_2]$, respectively, and with the other parameters fixed, one can show that the ODE system of Eq.~\eqref{modeloSIR} admits existence and uniqueness of solution in the time interval $[0,T]$.

\subsection{Well posedness}

In this subsection we prove that,  for a given time $T>0$, 
%and non-negative constants  $\kappa_1<\kappa_2$, $\rho_1 < \rho_2$ and $\lambda_1<\lambda_2$, 
for  each initial value $(S_0,I_0,R_0,M_0)$, with  $\kappa(t)$, $\rho(t)$ and $\lambda(t)$ varying continuously in the bounded interval $[0,T]$ %$[\kappa_1, \kappa_2]$,  $[\rho_1, \rho_2]$, $[\lambda_1,\lambda_2]$, respectively, 
and with the
other parameters fixed, one can show that the ODE system of Eq.~\eqref{modeloSIR} admits existence and uniqueness of solution in the time interval $[0,T]$. Furthermore, we prove that the solutions are continuously dependent on the initial data and on the parameters $\kappa$, $\rho$ and $\lambda$.

To prove the existence and uniqueness of solution of Eq. ~\eqref{modeloSIR},  in the Euclidean space $\mathbb{R}^{4}$, it is sufficient to verify that  the function
given by the right hand side of Eq. ~\eqref{modeloSIR} is Lipschitz continuous  with respect to spatial variable   (see \cite{Hale}).

For this, let $\xi(t)=(S(t),I(t),R(t),M(t))$ be and define $g:\mathbb{R}\times \mathbb{R}^{4} \to \mathbb{R}^{4}$ as
\begin{equation}
    g(t, \xi)=(g_1(t, \xi), g_2(t, \xi), g_3(t, \xi),g_4(t, \xi)), \label{F}
\end{equation}
where the $g_j: \mathbb{R} \times\mathbb{R}^{4} \to \mathbb{R}$ are the coordinate functions of $g$ given by
$$
g_1(t, \xi)=\nu(S+I+R)-\mu S-\kappa(t) SI,
$$
$$
g_2(t, \xi)=-\left(\mu +\frac1\tau\right) I+\kappa(t) SI,
$$
$$
g_3(t, \xi)=\rho(t) I-\mu R,
$$
and
$$
g_4(t, \xi)=\lambda(t) I,
$$
with $\kappa, ~ \rho, ~\lambda: [0, T] \to \mathbb{R}_{+}$ continuous functions.
\begin{proposition}
The function given in (\ref{F}) is Lipschitz continuous  with respect to the second variable.\label{Lip}
\end{proposition}
\begin{proof}
 Initially we will denote 
 $$
 \kappa_{max}=\sup_{t\in [0, T]} \kappa(t)~ \mbox{and}~ \rho_{max}=\sup_{t\in [0,~ T]}\rho(t)
 $$ 
 and we will consider
 $\mathbb{R}^{4}$ with the sum norm, that is, for $\xi=(S,I,R,M)$,
 $$
 \|\xi\|=|S|+|I|+|R|+|M|,
 $$
 where $|S|\leq S_{max}$, $|I|\leq I_{max}$, $|R|\leq R_{max}$ and $|M|\leq M_{max}$.
 
 %by denoting
  %$$  {S_1}_{max}=\sup_{t\in [0,~T]}S_1(t), ~ {I_2}_{max}=\sup_{t\in [0,~T]}I_2(t), $$
 Hence, it is easy to see that
\begin{eqnarray*}
|g_1(t,~\xi_1)-g_1(t,~\xi_2) &\leq& |\nu - \mu||S_1-S_2|+\nu|I_1-I_2|+\nu|R_1-R_2|+\kappa(t)|S_1I_1-S_2I_2|\\
&\leq& |\nu - \mu||S_1-S_2|+\nu|I_1-I_2|+\nu|R_1-R_2|+ \kappa(t)|S_1||I_1-I_2|\\
&+&\kappa(t)|I_2||S_2-S_1|\\
&\leq& |\nu - \mu||S_1-S_2|+\nu|I_1-I_2|+\nu|R_1-R_2|+ 
\kappa_{max} S_{max} |I_1 -I_2|\\
&+& \kappa_{max}~ I_{max}|S_1 -S_2|\\
&\leq& \left(|\nu - \mu|+\kappa_{max} I_{max}\right)|S_1 -S_2| 
+ (\nu+\kappa_{max} S_{max} )|I_1 -I_2|\\
&+& \nu |R_1-R_2|.
\end{eqnarray*}
Hence, by writing $L_1=\sup\{|\nu - \mu| +\kappa_{max} I_{max} , ~~ \nu  +\kappa_{max} S_{max} \}$, we obtain
\begin{eqnarray*}
|g_1(t,~\xi)-g_1(t,~\xi)| &\leq&  L_1 \left(|S_1-S_2|+|I_1-I_2|+|R_1 -R_2|\right)\\
&\leq&  L_1\|(\xi)-(\xi_2)\|.
\end{eqnarray*}
Similarly
\begin{eqnarray*}
|g_2(t, ~\xi_1)-g_2(t,~\xi_2)|&\leq& \left(\mu + \frac{1}{\tau} \right)|I_2-I_1|+ \kappa(t)|S_1I_1-S_2I_2|\\
&\leq& \left(\mu + \frac{1}{\tau} \right)|I_2-I_1|+\kappa(t)|S_1||I_2-I_1|+ \kappa(t)|I_2|S_1-S_2|\\
&\leq& \kappa_{max} ~I_{max}|S_1-S_2| +\kappa_{max}~S_{max}|I_1-I_2|+ (\mu +\frac{1}{\tau})|I_2-I_1|;
\end{eqnarray*}
thus, writing $L_2=\max \{\kappa_{max} I_{max}, ~~\kappa_{max}~S_{max} +  ~\mu +\frac{1}{\tau}\}$ we have
\begin{eqnarray*}
|g_2(t,~\xi_1)-g_2(t, \xi_2)| &\leq& L_2 \left(|S_1-S_2|+|I_1-I_2| \right)\\
&\leq& L_2\|\xi_1-\xi_2)\|;
\end{eqnarray*}
\begin{eqnarray*}
|g_3(t,~\xi_1)-g_3(t,\xi_2)| &\leq& \rho(t)|I_1-I_2|+\mu |R_2-R_1|\\
&\leq&\rho_{max}|I_1-I_2|+ \mu|R_1-R_2|\\
&\leq& L_3 \left(|I_1-I_2|+ |R_1-R_2|\right)\\
&\leq& L_3 \|\xi_2-\xi_2\|,
\end{eqnarray*}
where $L_3=\max\{\rho_{max}, ~\mu\}$,
and, 
$$
|g_4(t,~\xi_1)-g_4(t,~\xi_2)| \leq \lambda(t)|I_1-I_2|.
$$
Since $\frac{1}{\tau}=\rho(t)+\lambda(t)$, it follows that $\lambda(t)= \frac{1}{\tau}-\rho(t)$. Thus  $$
 \lambda_{max}=\sup_{t\in [0, ~T]} \left(\frac{1}{\tau}-\rho(t)\right) < \tau^{-1}.
 $$
 Hence
$$
|g_4(t,~\xi_1)-g_4(t,~\xi_2)| \leq \frac{1}{\tau}|I_1-I_2|\leq \frac{1}{\tau}  \|\xi_1-\xi_2\|.
$$

Therefore,
$$
\|g(t,\xi_1)-g(t,\xi_2)\| \leq L \|\xi_1-\xi_2\|,
$$
where $L=\max\{L_1, ~L_2, ~L_3, ~\frac{1}{\tau}\}$. This concludes the proof.
\end{proof}

\begin{remark}
From Propositin \ref{Lip} and  Picard-Lindelöf Theorem  \cite{Hale}, it follows that for  each initial value $\xi_0=(S_0,I_0,R_0,M_0)$, with  $\kappa(t)$, $\rho(t)$ and $\lambda(t)$ varying continuously in the bounded intervals $[0,T]$, the ODE system of Eq.~\eqref{modeloSIR} admits existence and uniqueness of solution in the time interval $[0,T]$, which is given, for $t\in [0,T]$, by
\begin{equation}
\xi(t)=\xi_0
+ \int_{0}^{t}g(s, \xi(s))ds.\label{integral}
\end{equation}
Furthermore, using (\ref{integral}) and Gronwall inequality  \cite{Hale}, we obtain the sensitivity with respect to the initial data. 

Indeed, denoting by $\xi(t,\xi_0)$ the solution of Eq.  (\ref{modeloSIR}) that at $t=0$ is $\xi_0$, we have
\begin{eqnarray*}
\|\xi(t, \xi_0^1)-\xi(t, \xi_0^2)\|&\leq& \|\xi_0^1 - \xi_0^2\| + \int_{0}^{t}\|g(s, \xi(s, \xi_0^1))-g(s, \xi(s, \xi_0^2))\|ds\\
&\leq& \|\xi_0^1 - \xi_0^2\|+ \int_{0}^{t}L\|\xi(s, \xi_0^1)-\xi(s, \xi_0^2)\|ds.
%\label{integral}
\end{eqnarray*}
Hence, from Gronwall inequality  
\begin{eqnarray*}
\|\xi(t, \xi_0^1)-\xi(t, \xi_0^2)\|&\leq& \|\xi_0^1 - \xi_0^2\|e^{ \int_{0}^{t}Lds}\\
&\leq& e^{T}\|\xi_0^1 - \xi_0^2\| \to 0,~\mbox{ as} ~ \|\xi_0^1 - \xi_0^2\| \to 0.
%\label{integral}
\end{eqnarray*}
\end{remark}

\subsection{Continuous dependence on parameters}
In this subsection we investigate the continuous dependence of the solutions of ~\eqref{modeloSIR} with respect to the variation of  its main parameters 
%In this subsection we prove that the solutions of ~\eqref{modeloSIR} depend continuously on variations of  the main parameters present in the model \cite{Diekmann2005}, \cite{Silva2020}.

\begin{proposition} Under the same  assumptions and notation from Proposition \ref{Lip}, the solution $(S(t),I(t),R(t),M(t))$   is continuous with respect to parameters $\kappa(t)$, $\rho(t)$ and $\lambda(t)$ present in Eq. \eqref{modeloSIR}.
\end{proposition}

\begin{proof}
Define $\theta: [0,T] \to \mathbb{R}^3$ by  $$
\theta(t)=(\kappa(t),\rho(t), \lambda(t)).
$$ 
Consider the metric defined on space of the continuous functions of $[0,~ T]$ in $\mathbb{R}^{3}$, 
$\mathbb{C}([0,T], \mathbb{R}^{3})$,
given by 
$$
dist(\theta, \theta_0)= \sup_{t\in [0,~T]}|\kappa(t)-\kappa_0(t)|+ \sup_{t\in [0,~T]}|\rho(t)-\rho_0(t)| + \sup_{t\in [0,~T]}|\lambda(t)-\lambda_0(t)|.$$ 
Now, we denote by $\xi_{\theta}(t)$ the solutions of Eq. (\ref{modeloSIR}) with respect to parameter $\theta(t)=(\kappa(t),\rho(t), \lambda(t))$ such that $\xi_{\theta}(0)=\xi_0$ and we denote by $\xi_{\theta_{0}}(t)$ the solutions with respect to parameter $\theta_{0}(t)=(\kappa_{0}(t),\rho_{0}(t), \lambda_{0}(t))$ such that $\xi_{\theta_{0}}(0)=\xi_0$. 

We need to show that
\begin{eqnarray*}
\|\xi_{\theta}(t) - \xi_{\theta_{0}}(t)\| \to 0, ~\mbox{as}~ ~dist(\theta, \theta_0) \to 0.
\end{eqnarray*}
For this, note that

\begin{eqnarray}
\|\xi_{\theta}(t) - \xi_{\theta_{0}}(t)\| \leq    
\int_{0}^{t} \|g(s,~\xi_{\theta}(s)) - g(s,~\xi_{\theta_{0}}(s))\|ds. \label{est_g}
\end{eqnarray}
Now
\begin{eqnarray*}
|g_1(s,~\xi_{\theta}(s)) - g_1(s,~\xi_{\theta_{0}}(s)))| &=& |\nu [S(s)-S_0(s)] + \nu [I(s)-I_0(s)]+ \nu [R(s)-R_0(s)]\\
&+& \mu[S_0(s)-S(s)] + \kappa(s)S(s)I(s) -\kappa_0(s)S_0(s)I_0(s)| \\
&\leq& (\nu +\mu) |S(s)-S_0(s)| + \nu |I(s)-I_0(s)|+ \nu |R(s)-R_0(s)|\\
&+& |\kappa(s)S(s)I(s) -\kappa_0(s)S(s)I(s)|\\
&+&|\kappa_0(s)S(s)I(s)- \kappa_0(s)S_0(s)I_0(s)|.%\\
\end{eqnarray*}
But
\begin{eqnarray*}
|\kappa(s)S(s)I(s) -\kappa_0(s)S(s)I(s)|= |S(s)I(s)||\kappa(s)-\kappa_0(s)|%\leq S_{max}~I_{max} |\kappa(s)-\kappa_0(s)|
\end{eqnarray*}
and
\begin{eqnarray*}
|\kappa_0(s)S(s)I(s)- \kappa_0(s)S_0(s)I_0(s)|&=& |  \kappa_0(s)|S(s)I(s)-S_0(s)I_0(s)|\\
&\leq& \kappa_0(s)\left(|S(s)||I(s)-I_0(s)|+|I_0(s)||S(s)-S_0(s)|\right).\\
%&\leq& {\kappa_0}_{max}(\left S_{max} |I(s)-I_0(s)|+{I_0}_{max}|S(s)-S_0(s)|\right).
\end{eqnarray*}
Hence
\begin{eqnarray*}
|g_1(s,~\xi_{\theta}(s)) - g_1(s,~\xi_{\theta_{0}}(s))| 
%&\leq& (\nu +\mu) |S(s)-S_0(s)| + \nu |I(s)-I_0(s)|+ \nu |R(s)-R_0(s)|\\
%&+& |S(s)I(s)||\kappa(s)-\kappa_0(s)| + \kappa_0(s)|S(s)I(s)-S_0(s)I_0(s)|\\
%&\leq& (\nu +\mu) |S(s)-S_0(s)| + \nu |I(s)-I_0(s)|+ \nu |R(s)-R_0(s)|\\
%&+& |S(s)I(s)||\kappa(s)-\kappa_0(s)|\\
%&+& \kappa_0(s)\left(|S(s)||I(s)-I_0(s)|+|I_0(s)||S(s)-S_0(s)|\right)\\
&\leq& (\nu +\mu) |S(s)-S_0(s)| + \nu |I(s)-I_0(s)|+ \nu |R(s)-R_0(s)|\\
&+& S_{max} ~I_{max}\sup_{s\in [0,T]}|\kappa(s)-\kappa_0(s)|\\
&+& {\kappa_0}_{max}\left(S_{max}~|I(s)-I_0(s)|+{I_{0}}_{max}||S(s)-S_0(s)|\right)\\
&\leq& (\nu +\mu + {\kappa_0}_{max} ~{I_{0}}_{max}) |S(s)-S_0(s)|\\
&+&  (\nu + {\kappa_0}_{max} ~S_{max}) |I(s)-I_0(s)|
+ \nu |R(s)-R_0(s)|\\
&+& S_{max}I_{max} ~dist(\theta, \theta_0).
\end{eqnarray*}
Thus, by writing $G_1=\{\nu +\mu + {\kappa_0}_{max} ~{I_{0}}_{max}, ~\nu + {\kappa_0}_{max} ~S_{max}\}$, we obtain
\begin{eqnarray}
|g_1(s,~\xi_{\theta}(s)) - g_1(s,~\xi_{\theta_{0}}(s)))|\leq G_1 \|\xi_{\theta}(s)- \xi_{\theta_{0}}(s)\| + S_{max}I_{max} ~dist(\theta, \theta_0).\label{est_G1}
\end{eqnarray}
\begin{eqnarray}
|g_2(s,~\xi_{\theta}(s)) - g_2(s,~\xi_{\theta_{0}}(s))|
&=& \left|\left(\mu + \frac{1}{\tau}\right)(I_0-I) + \kappa(s)S(s)I(s)-\kappa_0(s) S_0(s)I_0(s)\right|\nonumber\\
&\leq& \left(\mu + \frac{1}{\tau}\right) |I_0-I|+ |S(s)I(s)||\kappa(s)-\kappa_0(s)|\nonumber\\
&+&|\kappa_0(s)||S(s)I(s)-S_0(s)I_0(s)|\nonumber\\
&\leq& \left(\mu + \frac{1}{\tau}\right) |I_0-I|+ S_{max} ~I_{max}|\kappa(s)-\kappa_0(s)|
\nonumber\\
&+& {\kappa_0}_{max} \left(|S(s)||I(s)-I_0(s)| + |I_0(s)||S(s)-S_0(s)|\right)\nonumber\\
&\leq& {\kappa_0}_{max} ~{I_0}_{max}|S(s)-S_0(s)|+\left(\mu + \frac{1}{\tau} + {\kappa_0}_{max} ~S_{max}\right) |I_0-I|\nonumber\\
&+& S_{max} ~I_{max} ~dist(\theta,~ \theta_0)\nonumber\\
&\leq& G_2 \|\xi_{\theta}-\xi_{\theta_0} \| + S_{max} ~I_{max} ~dist(\theta,~ \theta_0), \label{est_G2}
\end{eqnarray}
where $G_2=\max\{{\kappa_0}_{max} ~{I_0}_{max}, ~ \mu + \frac{1}{\tau} + {\kappa_0}_{max} ~S_{max}\}$.

\begin{eqnarray*}
|g_3(s,~\xi_{\theta}(s)) - g_3(s, \xi_{\theta_{0}}(s))| &=& |\rho(s)I(s)-\rho_0(s)I_0(s) + \mu[R_0(s)-R(s)]|\nonumber\\
&\leq& |I(s)| |\rho(s)-\rho_0(s)| + |\rho_0(s)||I(s)-I_0(s)|\nonumber\\
&+&\mu[R_0(s)-R(s)]|\nonumber\\
&\leq& I_{max}\sup_{s\in [0,~T]}|\rho(s)-\rho_0(s)|+\left(\sup_{s\in [0,~T]}\rho_0(s)\right)|I(s)-I_0(s)|\nonumber\\
&+&\mu|R_0(s)-R(s)|\nonumber\\
&\leq&  I_{max}dist(\theta, \theta_0) + {\rho_{0}}_{max}|I(s)-I_0(s)|\nonumber\\
&+&\mu|R_0(s)-R(s)|\nonumber\\
&\leq& G_3 \|\xi_{\theta}(s)-\xi_{\theta_{0}}(s)\|+I_{max} ~dist(\theta, \theta_0),\label{est_G3}
\end{eqnarray*}
where $G_3=\max\{{\rho_{0}}_{max},~\mu
\}$. Now,
\begin{eqnarray*}
|g_4(s,~\xi_{\theta}(s)) - g_4(s,~\xi_{\theta_{0}}(s))| &=& |\lambda(s)I(s)-\lambda_0(s)I_0(s)|\nonumber\\
&\leq& |I(s)||\lambda(s)-\lambda_0(s)|+ \lambda_0(s)|I(s)-I_0(s)|\nonumber\\
&\leq&I_{max}~\left(\sup_{s\in [0,T]}|\lambda(s)-\lambda_0(s)|\right) + \left(\frac{1}{\tau}-\rho_0(s)\right) ~|I(s)-I_0(s)|\nonumber\\
&\leq&  \frac{1}{\tau} ~\|\xi_{\theta}(s)-\xi_{\theta_{0}}(s)\| +I_{max}~dist(\theta, \theta_0).\label{est_g4}
\end{eqnarray*}
Thus, by using (\ref{est_G1}), (\ref{est_G2}), (\ref{est_G3}) and (\ref{est_g4}) in (\ref{est_g}), we have
\begin{eqnarray*}
\|\xi_{\theta}(t) - \xi_{\theta_{0}}(t)\| &\leq& \int_{0}^{t}\{ 2(S_{max}I_{max} + I_{max})dist(\theta, \theta_0)+ (G_1+G_2+G_3+ \frac{1}{\tau}) \|\xi_{\theta}(s) - \xi_{\theta_{0}}(s)\|\}ds\\
&\leq& 2T(S_{max}I_{max} + I_{max})dist(\theta, \theta_0)+\int_{0}^{t} \left(G_1+G_2+G_3+ \frac{1}{\tau}\right)\|\xi_{\theta}(s) - \xi_{\theta_{0}}(s)\|ds.
\end{eqnarray*}
Therefore, from Gronwall Lemma, it follows that
\begin{eqnarray*}
\|\xi_{\theta}(t) - \xi_{\theta_{0}}(t)\| &\leq& 2T(S_{max}I_{max} + I_{max})dist(\theta, \theta_0)e^{\left(G_1+G_2+G_3+ \frac{1}{\tau}\right)T} \to 0,
\end{eqnarray*}
 as $dist(\theta, \theta_0)$.
 \end{proof}

\section{Estimates for the parameters }

In this section, we explain  some.
estimates for the parameters present in the model.

\subsection{Reproduction number}
It is of paramount importance to know if a contagious
disease will become epidemic or not in a population.
It is also important to know when  it will be possible
to control an epidemic, that is, when it will be possible
to block its growth.
This will happen when $\frac{dI}{dt}\leq 0$.
From Eq.~\eqref{modeloSIR}, we verify that this condition
is equivalent to
\begin{equation}
    -\left(\mu+\frac1\tau\right)+\kappa S(t)\leq0\implies \frac{\kappa S(t)}{\mu+\frac1\tau}\leq 1.
\end{equation}
In the beginning of the epidemic we obtain that the value
of the following ratio
\[
R_0=\frac{\kappa S_0}{\mu+1/\tau},
\]
known in the literature \cite{murray2007mathematical} as
the basic reproduction number, is what indicates 
whether we will have an epidemic or not.
When $R_0>1$, the disease will spread,
whereas when $R_0 <1$ the contagion loses strength and the dissemination of the virus will be controlled.
In our case $S_0=1$ and $S(t)\leq 1$, thus
at any time after the onset of the epidemic
the disease will stop spreading when 
\begin{equation}
    R_0(t)=\frac{\kappa(t) S(t)}{\mu+1/\tau}\leq 1.
\end{equation}
We have that $R_0<1$ is a sufficient condition that the
epidemic will enter remission, but in general it is not
a necessary condition.
The necessary condition is that $R_0(t)<1$.
Although,  we are just over 5 months into the
current epidemic in Brazil at the time of writing, $M(t)<<1$,
$S(t)\approx1$,
the critical condition is still approximately $R_0=1$ and
the critical value of $\kappa(t)$ is $\kappa^*=\mu+1/\tau$.

As there is no efficacious treatment against 
COVID-19 at the time of writing this paper to the authors' knowledge, it is not yet possible to easily alter the average
time of infection $\tau$. 
As $1/\tau>>\mu$, 
thus the only viable manner of decreasing $R_0(t)\approx\kappa(t)\tau$ is by reducing the value of $\kappa(t)$, which can be obtained with social distancing measures and the use of PPEs.
\subsection{Probabilistic model}
\label{let_rec}
It is fundamental that we have good estimates for the
recovery and the lethality rates.
In order to obtain these estimates we will use a very
simple probabilistic model.
Suppose that a person is infected at a given instant $n$ (which may be a day, an hour, or a minute for example), then the probability that the infected remains sick until the following
instant $n+1$ is $q$, the probability that the infected recovers
in the next instant is $p$,
and the probability that the infected dies is
$s$, in such a way that $p+q+s=1$. 
Here, we suppose that $q$ and $p+s$ remain constant during the course of the disease.
Hence, we have the following probability table
\begin{table}[ht]\caption{Probabilistic model} 
\label{tabp}
\begin{tabular}{c|ccccccc}\hline
Situation $\setminus$ Instant & 0 & 1  & 2 & \ldots & $n$& \ldots\\\hline
Recovered& $p$ & $qp$ & $q^2p$ &  \ldots &$q^np$ &\ldots\\
Death& $s$ & $qs$ & $q^2s$ &  \ldots &$q^ns$ &\ldots\\\hline
\end{tabular}
\end{table}
\FloatBarrier
Note that the normalization
\[
\sum_{n=0}^\infty q^np+\sum_{n=0}^\infty q^ns=\frac{p+s}{1-q}=1,
\]
implies that this probabilistic model is well defined.

If $n$ is sufficiently large, only two outcomes are possible: either the infected individual recovers or dies. 
Hence, based on the Table \ref{tabp} we find that
the probability of recovery and of death are, respectively, 
\beq
\begin{aligned}
P_\rho &=\sum_{n=0}^\infty pq^n=p/(1-q),\\
P_\lambda &=\sum_{n=0}^\infty sq^n=s/(1-q).
\end{aligned}
\label{probRecDea}
\eeq
Using these probabilities we find that the average number
of instants (minutes, hours, days, etc) of the infection is given by
\beq
\bar n=\sum_{n=1}^\infty(p+s)nq^n=(p+s)F_1(q)=(1-q)\sum_{n=1}^\infty nq^n,
\eeq
in which the summation $F_1(q)=\sum_{n=1}^\infty nq^n$ 
can be calculated in the following form
\[
qF_1(q)=\sum_{n=1}^\infty nq^{n+1}=\sum_{n=2}^\infty (n-1)q^{n}
=\sum_{n=2}^\infty nq^{n}-\sum_{n=2}^\infty q^{n}
= F_1(q)-\sum_{n=1}^\infty q^{n}=F_1(q)-\frac q{1-q}.
\]
Hence, we obtain
\beq
F_1(q)= \frac{q}{(1-q)^2}.
\eeq
Therefore, we find that the average number of
instants of the infection is
\begin{equation}
\bar n = \frac q{1-q},
\end{equation}
from where we obtain that $q=\bar n/(1+\bar n)$.
We can also find that the average number of time intervals until recovery is given by
\[
\bar n_\rho= pF_1(q)=\frac{pq}{(1-q)^2}=P_\rho\bar n
\]
and the average number of time intervals until death is
\[
\bar n_\lambda=sF_1(q)=\frac{sq}{(1-q)^2}=P_\lambda \bar n.
\]
Note that $\bar n=\bar n_\rho+\bar n_\lambda$, that is, 
the average infection time span is the sum of the average time span for recovery and the
average time span until death.
If only these two processes were present, it would lead to the following difference
equation for the number of infected
\[
I(n+1)=I(n)-(p+s)I(n)=I(n)-(1-q)I(n)=I(n)-\frac1{1+\bar n}I(n).
\]
If we take $n$ to indicate the $n$-th time interval,
such as a minute, in which there is not much variation
in the quantities $S$, $I$, $R$, and $M$, 
hence, we can approximate
\[
\frac{dI}{dt}\approx \frac{\Delta I}{\Delta t}=-\frac{1}{(1+\bar n)\Delta t}I(t_n),
\]
in which $t_n=n\Delta t$.
In this work, we take $\Delta t=1~$hour$~ =1\rm{~day}/24$.
The average time span of infection can be obtained from the 
following equation
\beq
\frac1\tau=\lambda+\rho=\frac{1}{(1+\bar n)\Delta t}.
\eeq

Hence, we obtain the following expressions for the rates of
lethality and recovery
\begin{equation}
    \begin{aligned}
    \lambda &= \frac{\bar n_\lambda}{\bar n\tau}=\frac{P_\lambda}{\tau},\\
     \rho &=\frac{\bar n_\rho}{\bar n\tau}= \frac{P_\rho}{\tau}.
    \end{aligned}
    \label{lambda_rho}
    \end{equation}
\subsubsection{Standard deviation}
\label{sigma}
Here, we calculate the standard deviation for this probabilistic process in the number of time intervals $n$
of infection, from the contamination until recovery or death.
In order to achieve that, we have first to calculate the
sum
\begin{equation}
    F_2(q)=\sum_{n=1}^\infty n^2q^n.
\end{equation}
We can calculate this sum by noting that
\[
q^2F_2(q)=\sum_{n=2}^\infty(n-2)^2q^n=F_2(q)+4\sum_{n=1}^\infty(1-n)q^n=F_2(q)+\frac{4q}{1-q}-q-4F_1(q).
\]
Hence, we find
\begin{equation}
    F_2(q)=\frac{q(1+q)}{(1-q)^3}.
\end{equation}
We then obtain
\[
\overline{n^2}=(1-q)F_2(q)=\frac{q(1+q)}{(1-q)^2}
\]
and 
\[
\overline{n}^2=\frac{q^2}{(1-q)^2}.
\]
We can now write the standard deviation of $n$ as
\begin{equation}
    \sigma=\sqrt{\overline{n^2}-\overline{n}^2}
    =\frac{\sqrt{q}}{1-q}.
\end{equation}
This shows that the statistical fluctuations in the time duration of the infection as $q$ grows to 1.
By reducing $q$ one not only decreases $\overline{n}$ but also $\sigma$.
\subsection{Statistical prediction of active, recovery, and death cases}
\label{stat_predict}
It is very important for government officials to have estimates of
the number of active (infected) cases in a given population, since
these people are the ones who spread the disease.
Based on  knowledge of new daily infections (new confirmed cases), which we can obtain from COVID-19 databases that are updated at least every day with
data from across the globe, we can make several estimates.
One simple estimate of the number of active cases at time
$t_k$ is 
\beq
A(t_k)=C(t_k)-\mr(t_k)-D(t_k)\approx C(t_k)-C(t_k-\tau),
\label{active_simple}
\eeq
where $C(t_k)$ is the number confirmed cases, 
$\mr(t_k)$ is the number of recovered cases, $D(t_k)$ represents the accumulated number of death cases at time $t_k$.
Another approach is to make a statistical estimate for the number of active cases at time $t_k$ that follows the probabilistic process described in Sec.~\ref{let_rec}, which is
the following
\beq
A(t_k)=-\sum_{i=-1}^k q^{i}P_0\Delta S(t_{k-i})=
\sum_{i=0}^k q^{i}\Delta C(t_{k-i}),
\label{active_stat}
\eeq
where $A(t_k)=P_0I(t_k)$, $\Delta S(t_n)=S(t_n)-S(t_{n-1})$ and, likewise,
$\Delta C(t_n)=C(t_n)-C(t_{n-1})$.
The index $k$ starts at 0.
We also have the following probabilistic estimates for the
daily increase of recovered and death cases at the $k$-th day
\begin{align}
    \Delta \mr(t_k) &= \sum_{i=0}^{k-1} q^{i}p\Delta C(t_{k-1-i}),\\
    \Delta D(t_k) &= \sum_{i=0}^{k-1} q^{i}s\Delta C(t_{k-1-i}),
    \label{rec_death}
\end{align}
These probabilistic estimates complement the predictions from the epidemiological dynamical system model.
The probabilistic aspect of the pandemic evolution
can be most clearly seen when one considers short time variations,
such as daily data on the number of new confirmed, recovered, or
death cases.
Also this allows us to make estimates of the evolution of 
the active, recovered, and death cases solely from the confirmed
cases time series.
This is most important where the recovered cases are not available
or are less reliable, specially when one is counting the recovery
of outpatients, since they may not take a second test to show
they are actually free of the infection.
%Also these estimates may help point out if local governments are inflating the numbers of recovery cases, possibly motivated by political or economical interests, specially with the intent of speeding up the reopening of businesses.
\subsection{Estimates of the parameters used in the model}
We now make some estimates for the parameters
$\nu$, $\mu$, $\rho$, and $\lambda$ present in the dynamical system given in Eq.~(\ref{modeloSIR}).

\subsubsection{Birth and death rates}
In order to make our model more precise, we obtained
the annual birth rate (ABR) and the annual mortality rate (AMR)
from the most recently published data from Germany and from Brazil before the spread of the pandemic.

We also converted these annual rates into daily rates using the geometrical progression formulas
\[
\begin{aligned}
(1+\nu)^{365} &=1+ABR,\\ 
(1-\mu)^{365} &=1-AMR.
\end{aligned}
\]
Hence, we find the daily birth and mortality rates
\beq
\begin{aligned}
    \nu
    &=(1+ABR)^{1/365}-1=e^{\frac1{365}\ln(1+ABR)}-1\\
    %\approx4.34626\times10^{-5}/day,\\
    \mu
    &=1-(1-AMR)^{1/365}=1-e^{\frac1{365}\ln(1-AMR)}
    %\approx1.9604\times10^{-5}/day.
\end{aligned}
\eeq
The German data on birth and death rates were obtained from the Federal Statistical Office \cite{destatisGeboren, destatisGesterben}.
The Brazilian data was obtained from the Brazilian Institute of Geography and Statistics (IBGE), which are from
2018.
For the birth rate we divided the number of born alive infants by the population estimate for 2018 and, likewise,
the number of deaths by the 2018 population estimate.
The data on the born alive was collected from the site
 \url{https://sidra.ibge.gov.br/tabela/2609}, the number of deaths from the site \url{https://sidra.ibge.gov.br/tabela/2654}, and
 the population estimate from the site \url{https://www.ibge.gov.br/estatisticas/sociais/populacao/9103-estimativas-de-populacao.html?edicao=22367&t=resultados}.
The birth and death rates used are shown in the table \ref{numuRates} below.
We provide table \ref{dailyDeaths} with information on the
average total daily deaths (without the born dead) before the pandemic based on annual death rates from 2018. This should be compared with the
average daily deaths due to COVID-19 as another
way of assessing the severity of this pandemic.
\begin{table}[ht]\caption{Birth and death rates} 
\label{numuRates}
\resizebox{\columnwidth}{!}{
\begin{tabular}{c|c|c|c|c}\hline
Population &  ABR (year$^{-1}$) & $\nu$ (day$^{-1}$)  & AMR (year$^{-1}$) & $\mu$(day $^{-1}$)\\\hline
Germany& 0.0095 &$2.591\times10^{-5}$ & 0.0115& $3.169\times10^{-5}$\\
Brazil& $0.0139$ & $3.7844\times10^{-5}$& 
$6.1560\times10^{-3}$ & $1.6918\times10^{-5}$\\
Paraíba & $0.0148$ & $4.0139\times10^{-5}$
& $6.5325\times10^{-3}$& $1.7956\times10^{-5}$\\
Campina Grande & $0.0158$ & $4.2943\times10^{-5}$ & $7.1342\times10^{-3}$ & $1.9616\times10^{-5}$\\
\hline
\end{tabular}
}
\end{table}
\begin{table}[ht]\caption{Pre-pandemic populations and average daily deaths}
\label{dailyDeaths}
\begin{tabular}{c|ccc}\hline
Location &  $P_0$ & $\mu$(day $^{-1}$) & average deaths (day$^{-1}$)\\\hline
Germany& 83,149,360 &$3.169\times10^{-5}$ &  2635\\
Brazil& 211,049,527 & $1.6918\times10^{-5}$&3570\\
Paraíba &  4,018,127&$1.7956\times10^{-5}$&72\\
Campina Grande &409,731 & $1.9616\times10^{-5}$&8\\
\hline
\end{tabular}
\end{table}
\FloatBarrier
\subsubsection{Contagion rate}
\label{conRate}
We use the following method to estimate the time-dependent contagion rate.
From Eq.~\eqref{modeloSIR}, we can write
\beq
\kappa(t) =\frac1{S(t)}\left[\frac1{I(t)}\frac{dI(t)}{dt}
+\mu +\frac1\tau\right].
\eeq
This can be approximated by
\beq
\kappa_k=\kappa(t_k) \approx\frac1{S(t_k)}\left[\frac1{A(t_k)}\frac{\Delta A(t_k)}{\Delta t}+\mu +\frac1\tau\right],
\label{kappa_est}
\eeq
where $\Delta t=1$day.
The daily approximation has too much fluctuation.
A better approach is to replace the daily derivative by the slope of a linear regression of 7 consecutive active cases
data points $\{A_k, A_{k+1}, ..., A_{k+6}\}$.
One rolls this week-long window over the entire set of data
points calculating $\kappa_k$.
For the last six days of data, we use a backwards window
with the data points $\{A_k, A_{k-1}, ..., A_{k-6}\}$.
We chose the week interval, because all epidemiological data
we have seen present weekly modulations.
The most difficult part of this estimating method of $\kappa(t)$
occurs at the beginning of the data sets, when the number
of active cases is very small.
From Eq.~\eqref{kappa_est}, one can see that for small values of $A(t_k)$,
any errors in the derivative approximation $\frac{\Delta A(t_k)}{\Delta t}$ are amplified.
Therefore, in most cases we discard the first days or weeks
of epidemiological data, usually up to the day the first death
occurred.
In some cases, we also had to apply a cutoff to eliminate the
highest values of $\kappa(t)$ right at the beginning of the time
series.
Apart from the numerical errors described above, notice that
any new imported infected 
case contributes strongly to the spread of the disease near
the outbreak of the epidemic.
Hence, at the beginning $\kappa(t)$ tends to be very high.
This also reflects the fact that at the start of the pandemic
most populations did not keep social distancing nor used PPEs
such as masks.
\subsubsection{Lethality and recovery rates}
\label{letRecRates}
We use the following method to estimate the time-dependent lethality rate.
Based on Eq.\eqref{rec_death}, we can write the probability
of dying for those infected with COVID-19 in a one day time
interval as
\beq
    s_k= \dfrac{\Delta D(t_k)}{\sum_{i=0}^{k-1} q^{i}\Delta C(t_{k-1-i})}.
\eeq
From this we obtain the daily lethality probability $P_\lambda(k)$ from Eq.~\eqref{probRecDea}.
Afterwards, from Eq.~\eqref{let_rec} we find the lethality rate $\lambda_k$ for the $k-$th day.
Consequently, from $\rho_k+\lambda_k=1/\tau$ we also find the
recovery rate $\rho_k$ at the $k-$th day.
The simpler approach of estimating $s_k$ with 
$\Delta D_k/A_{k-1}$ introduces far more fluctuations.
\subsection{Forecasting}
\FloatBarrier
We use a list with the last three weeks of $\kappa(t)$ data to calculate
the transition probabilities of a simple two-state Markov chain as shown
in Fig.~\ref{mchain}.
From this list, we obtain a list of differences
$\Delta\kappa_i=\kappa(t_i)-\kappa(t_{i-1})$.
Based one this, we make two probability distributions, one for increments in
$\kappa(t)$ and the other for decrements.
If $\Delta\kappa_i<0$, in the differences list, we replace it with 0, otherwise we
replace it with 1.
From this list of 0's and 1's, we obtain a list of lengths of the continuous
sequences of 0's and 1's.
We then obtain the average length of continuous stretches of 0's and 1's.
We call the average length of increments of the contagion rate  $\bar n_+$
and of decrements $\bar n_-$.
Based on this we find the following transition probabilities
\beq
\begin{aligned}
    q_{++} &= \frac{\bar n_+}{1+\bar n_+},\qquad p_{+-} &= 1-q_{++},\\
    q_{--} &= \frac{\bar n_-}{1+\bar n_-},\qquad p_{-+} &=  1-q_{--}.
\end{aligned}
\label{probs}
\eeq
Now we use the Markov chain to decide between increasing or decreasing $\kappa(t)$.
The values of the increments or decrements are taken from the two probability
distributions obtained as described above.
We then randomly obtain predicted values of $\kappa$ for two weeks based on
a three-week history of data.
We repeat this process for 1000 times and make an average of all these
trajectories as it is depicted in Fig.~\ref{casosConfirmadosBra}.
From these trajectories we also obtain a 95\% confidence interval.
In the figure, we compare actual data with our predictions for the two most
recent weeks.
Using this same approach, we can forecast the lethality rate based on
a list of daily variation of $\lambda(t)$.
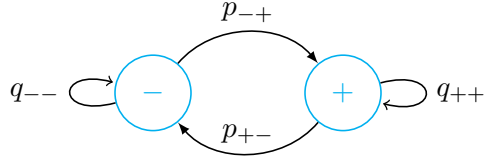
\begin{figure}[h!]
    \begin {center}
    \begin {tikzpicture}[-latex ,auto,node distance=2.5cm, on grid ,
    semithick ,
    state/.style ={ circle ,top color =white , bottom color = white ,
    draw,processblue , text=processblue , minimum width =1 cm}]
    \node[state] (A) at (0, 0) {$-$};
    \node[state] (B) [right =of A] {$+$};
    \path (A) edge [loop left =180] node[left] {$q_{--}$} (A);
    \path (A) edge [bend left=50] node[above =0.0 cm] {$p_{-+}$} (B);
    \path (B) edge [bend left =50] node[above] {$p_{+-}$} (A);
    \path (B) edge [loop right =50] node[right] {$q_{++}$} (B);
    \end{tikzpicture}
    \end{center}
    \caption{Markov chain diagram with transition probabilities. In the $-$ state $\kappa(t)$ is decreasing
    in time, while in the $+$ state, it is increasing.}
    \label{mchain}
\end{figure}
\FloatBarrier
\section{Results and Discussion}
The data of the number of confirmed, recovered and
 death cases from Germany and Brazil were obtained from the site <https://data.humdata.org/dataset/novel-coronavirus-2019-ncov-cases>, (accessed on 11/01/2020, with data collected until 10/31).
 We obtained the time series of confirmed and death cases of the State of Paraíba and the City of Campina Grande in the site:
\url{https://data.Brazil.io/dataset/covid19/_meta/list.html}
(accessed on 11/01/2020, with data collected until 10/31).

We used the Odeint function of the Python's scientific library package SciPy \cite{scipy1.0} to integrate the ODE system of Eq.~\eqref{modeloSIR} with the integration time-step $dt=1.0/24$, 
which corresponds to an hour when the time unit is a day.
In the cases investigated, we took $\tau$ 13-14~days.
We chose the value that provided the best fit of the active cases with the delay estimate given in Eq.~\eqref{active_simple} when the active cases data was available, such as in the cases of Germany and Brazil. 
When the active cases data was not available, we chose $\tau$
that would give the best fit between the theoretical
prediction for the active cases and the corresponding estimate
of Eq.~\eqref{active_simple}.
We suppose this parameter does not change appreciably during the time scale of the outbreak of the pandemic until now,
or at least until more efficient treatments are discovered.
 The initial values used are:
 $S(0)=1-C_0/P_0$,  $I(0)=A_0/P_0$, $R(0)=\mr_0/P_0$, and $M(0)=D_0/P_0$, with $P_0$ being the population just before the outbreak of the pandemic, $C_0$ is the initial number of confirmed cases, $A_0$ is the initial number of active cases, $\mr_0$ is the initial number of recovered cases, and $D_0$
 is the initial number of death cases, usually either 0 or 1.
 We now apply our model to the four cases of COVID-19 dissemination: in Germany, in Brazil, in the State of Paraíba, and in the City of Campina Grande.
 
In Fig.~ \ref{contagioObitosBRA}, we show results of numerical simulation for a range of values of $\kappa$.
Unlike the other results, $\kappa$ is held constant during each time integration of the equations of motion
given in Eq.~\eqref{modeloSIR}.
This result is important in conveying the message of
the paramount importance of the contagion rate on the
possible outcomes of the pandemic.
Not only we observe an increase in the number of deaths
when the contagion rates increase, but we also see a
sharp transition.
When there is a growth in the contagion rate from $0.1$ to $0.15$, this gives rise to an extremely sharp increase in the number of deaths.
This implies that there is a critical value of $\kappa$,
around which there is a rapid increase in the total
number of deaths due to the pandemic.
Beyond the critical value, we see a saturation in the
total number of deaths.
The value of $\kappa$ that corresponds to $R_0=1$
in our model is $\kappa=\mu+1/\tau=0.0714$.
% {\color{red}Such a low value indicates that the precautionary measures taken by the population are far below what they should be. }
We believe, this reinforces the great relevance of social distancing, since increasing the average distance between people, we will be decreasing the rate of contagion and, consequently, decreasing also the number
of deaths due to COVID-19.
\begin{figure}[ht]
    \centerline{\includegraphics[scale=0.8]{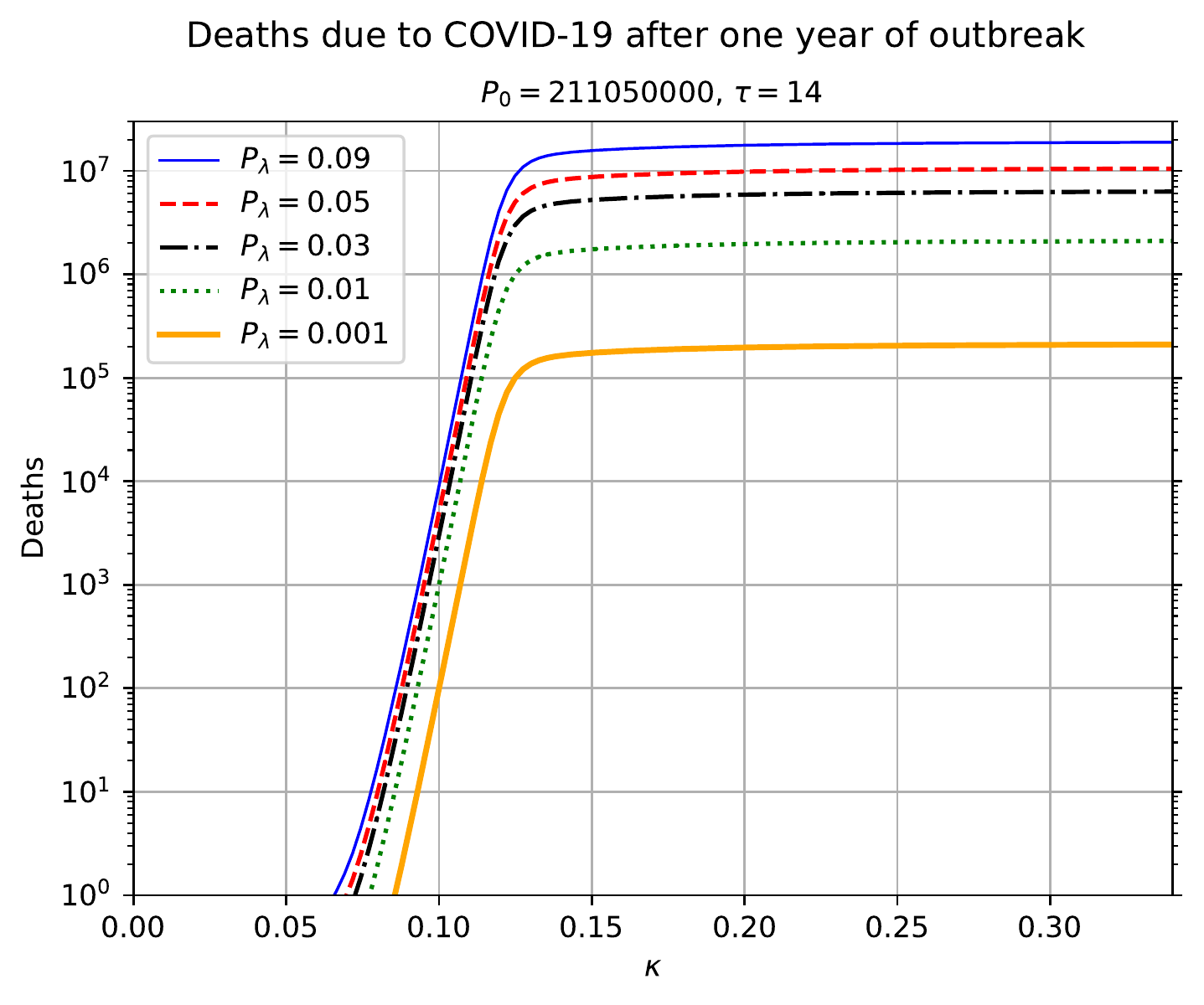}}
    \caption{Total number of accumulated deaths as a consequence of the epidemic as a function of the contagion rate $\kappa$.
    The total time duration for each value of  
     $\kappa$ is 365 days. The used parameters in this simulation are indicated above the figure. The value of $\kappa$ that
    corresponds to $R_0=1$ is $\kappa^*=\mu+1/\tau=0.0714$.
    }    
    \label{contagioObitosBRA}
 \end{figure}   
 \FloatBarrier
\subsection{Evolution of COVID-19 in Germany}
We consider the case of Germany as a benchmark test for
our epidemiological model.
This is so because it is widely believed that the cases
from Germany are better accounted for than in most other countries, with widespread testing of the population \url{https://www.labmate-online.com/news/laboratory-products/3/breaking-news/how-germany-has-led-the-way-on-covid-19-testing/52141}.
The initial population considered is $P_0=83.14936$ millions. 
The first contagion was registered on 01/27/2020 and the first
death due to COVID-19 was registered on 03/09/2020.
We chose this latter date as the initial point of our numerical integration.

In Fig.~\ref{confirmedCasesGer}{\bf A}, we show the official
data on the confirmed cases of COVID-19 plotted alongside the theoretical prediction, in which a very good agreement with the
epidemiological data was obtained.
The contagion rate function used is given in Fig.~\ref{kappa_tGer}.
In frame {\bf B}, we plot the recovered cases
data along with the theoretical predictions based on our model.
In frame {\bf C}, we plot the death cases with respective theoretical curve.
In frame {\bf D}, we plot the active cases
data along with the theoretical predictions based on our model.
The theoretical fit is not as good as in the previous
figures, but it is still quite reasonable.
We see a slow decline in the number of infected.
Once the total number of confirmed cases basically saturates,
the evolution of the active cases can be traced
with a purely statistical model as the one we developed
in Sec.~\ref{let_rec}. 
From the statistical point of view the slow decay of the active
cases has to do with the large value of the dispersion in the duration of the infection, as shown in Sec.~\ref{sigma}.
We also plot two estimates of the active cases obtained solely
from an analysis of the confirmed cases data.
The delay estimate of active cases at a given time $t$
is based on Eq.~\ref{active_simple}.
%just the difference between the number of confirmed cases at time $t$ minus the number of confirmed cases at time $t-\tau$.
The statistical estimate is obtained from the probabilistic
process given in Eq.~\ref{active_stat}.
This estimate came closer to the theoretical dynamical system model, what
shows its consistency with the probabilistic model of Sec.~\ref{let_rec}.
Both estimates came fairly close to the real data.
In the present case, they seem to bracket the real data.
Note also, that at the last two weeks of the time series we validate
the forecasting model based on the Markov chain.
Here, we show a 95\% confidence band along two weeks.
In this case the all epidemiological data fell within the predictions.

In Fig.~\ref{kappa_tGer}{\bf A}, we show the graph with the
parameter estimation for the contagion rate function $\kappa(t)$.
At the beginning of the time series, the contagion rate is
very high. To achieve the best fit between theory and the data
in Fig.\eqref{confirmedCasesGer}, we had a cutoff at $\kappa=0.41$, otherwise we used the estimation method described in Sec.~\ref{conRate} to obtain this time series.
Note that we observe weekly modulations, likely reflecting
the fact that in weekdays the contamination is higher than
in weekends.
The rapid decrease of $\kappa(t)$ intensifies approximately around 03/22/2020, when strict social distancing
rules were imposed by the German Government \cite{dw2020, dehning2020inferring}.
This shows that these measures were very efficient in
containing the spread of the epidemic.
In frame {\bf B}, we show the corresponding reproduction number
time series. Around April/09, the $R_0(t)$ becomes below 1
and stays there until about June/13, when for a period
of roughly 10 days it rose slightly above 1.
In frame {\bf C} we show the time-dependent lethality and
recovery rates, $\lambda(t)$ and $\rho(t)=1/\tau-\lambda(t)$, respectively.
These estimates were obtained using the method described in Sec. \ref{letRecRates}.
\FloatBarrier
\begin{figure}[ht]
    \centerline{\includegraphics[width=\textwidth]{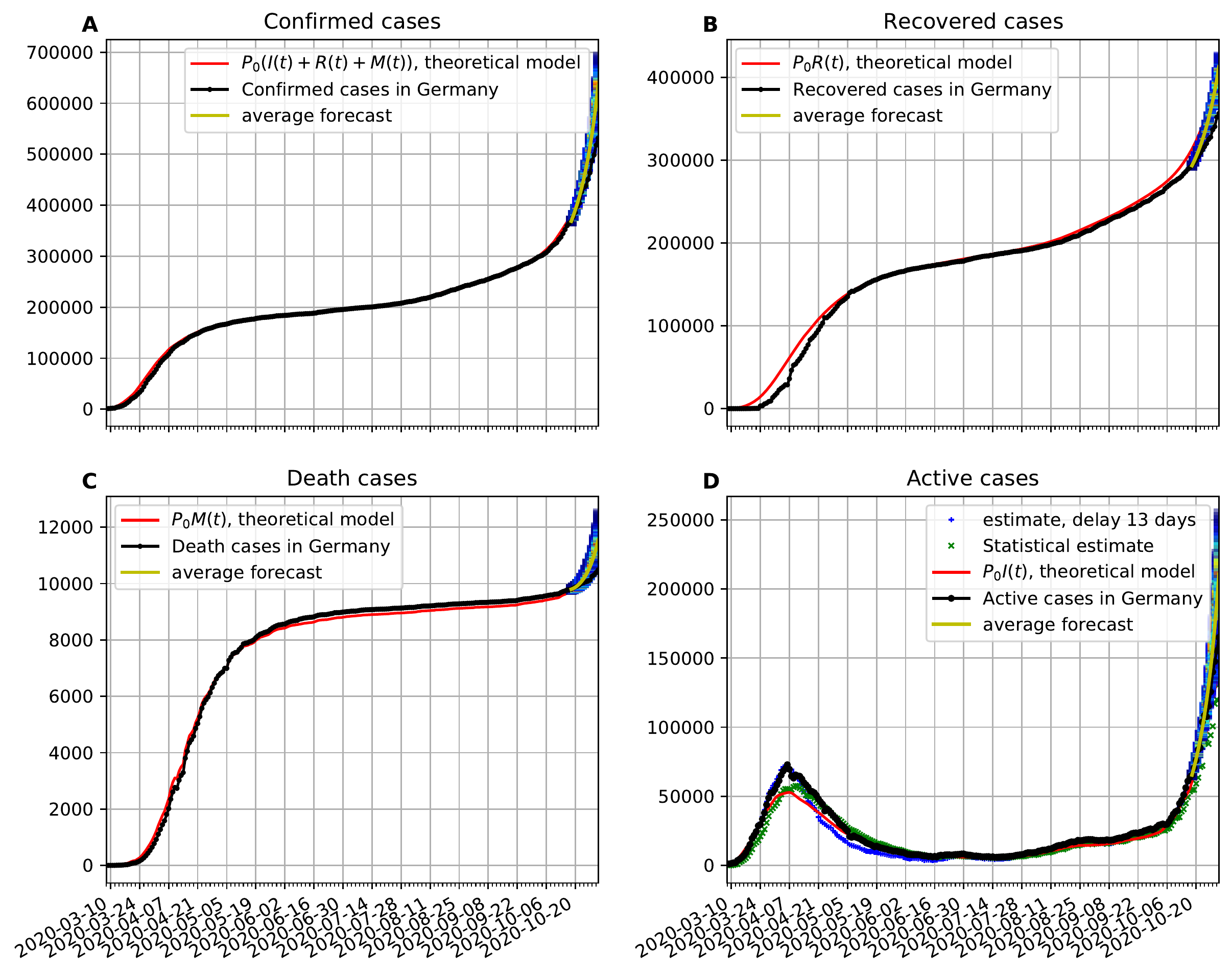}}
    %\centerline{\includegraphics[width=\textwidth]{confirmed_Germany.pdf}}
    \caption{{\bf A} Confirmed cases data compared with the theoretical prediction. 
     {\bf B} Recovered cases data in comparison with the theoretical prediction.
     {\bf C} Death cases data compared with the 
    theoretical prediction.
    {\bf D} Active cases compared with delay estimate, statistical estimate, and the predicted theoretical curve.
     In all cases the functions $\kappa(t)$, $\lambda(t)$, and $\rho(t)$ vary in time according to Fig.~\ref{kappa_tGer}.
    }
    \label{confirmedCasesGer}
   \end{figure} 
   
%   \begin{figure}[ht]
%     \centerline{\includegraphics[width=\textwidth]{deaths_Germany.pdf}}
%     \caption{
%     The number of official deaths compared with the 
%     theoretical prediction.
%     The theoretical fit was obtained with the same parameters of Fig.~\ref{confirmedCasesGer} and with the same contagion function $\kappa(t)$ of
%   Fig.~\ref{kappa_tGer}.
%     The best fit was obtained with the lethality probability 
%     $P_\lambda=0.05$ and the average time of infection $\tau=14$~days.
%     }
%     \label{deathsGer}
% \end{figure}

% \begin{figure}[ht]
%     \centerline{\includegraphics[width=\textwidth]{recovered_Germany.pdf}}
%     \caption{ The theoretical fit is obtained with the same parameters and with the same function $\kappa(t)$ that was used in Fig.~\ref{confirmedCasesGer}.
%     }
%     \label{recoveredGer}
%     \end{figure}
    
    % \begin{figure}[ht]
    % \centerline{\includegraphics[width=\textwidth]{active_Germany.pdf}}
    % \caption{ The theoretical fit is obtained with the same parameters and with the same function $\kappa(t)$ that was used in Fig.~\ref{confirmedCasesGer}.
    % }
    % \label{activeGer}
    % \end{figure}
    
  \begin{figure}[ht]
    \centerline{\includegraphics[width=0.6\textwidth]{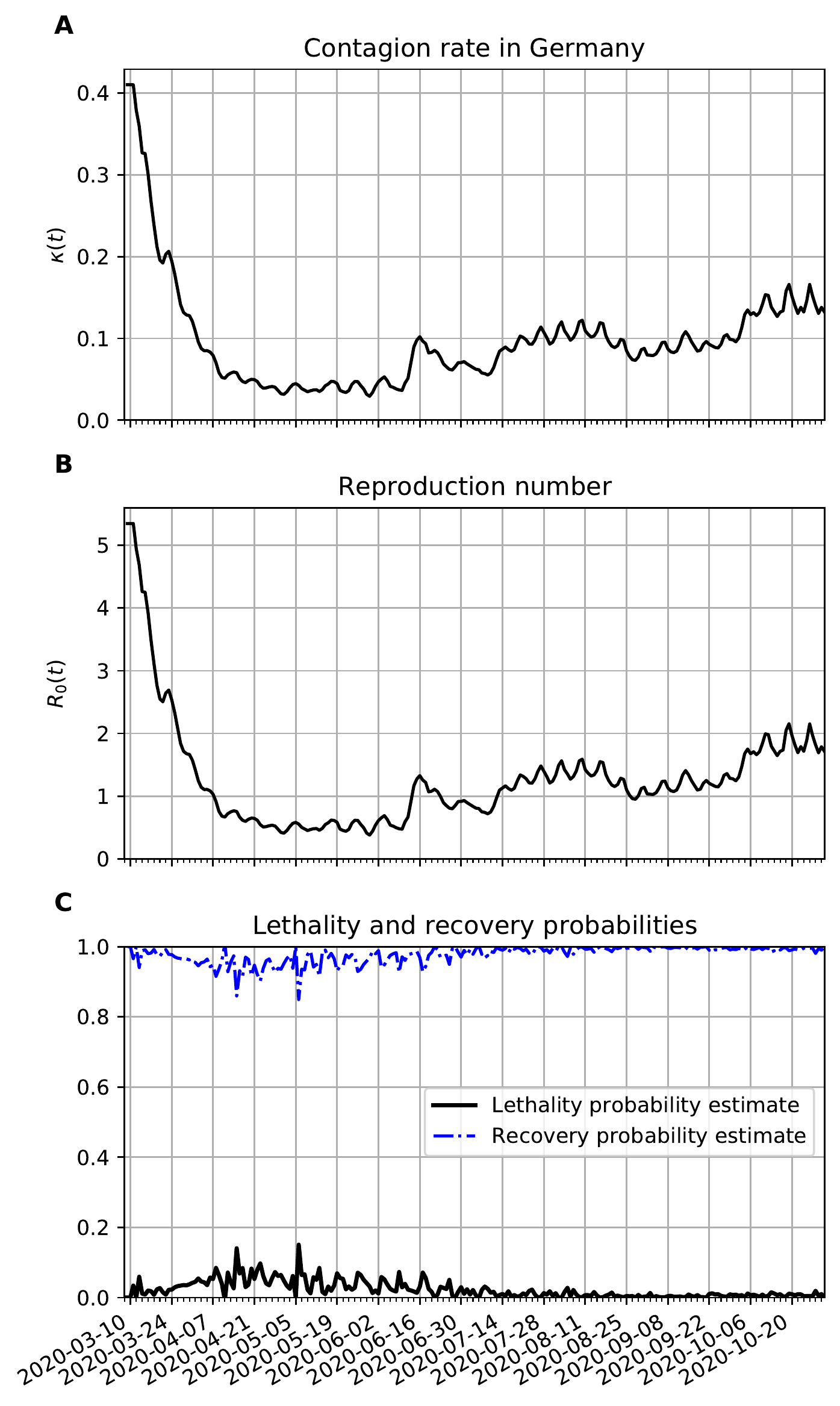}}
    \caption{Estimation of parameters of the epidemiological model. 
    The initial date of the time series is the date in which the first death due to COVID-19 occurred in Germany.
    {\bf A} the time variation of the contagion rate. 
     The initial rate was very high, possibly due to exogenous cases in the early stages of the pandemic.
    To better fit the data with the predictions of our model
    we had to truncate $\kappa(t)$ with values above 0.41 at
    the beginning of the series.
    {\bf B} the reproduction number. {\bf C} the lethality and recovery rates. 
    }
    \label{kappa_tGer}
    \end{figure}
  
   \FloatBarrier
\subsection{Evolution of COVID-19 in Brazil}
We consider the initial time the day of the first
death case in Brazil.
 We take Brazil's population to be approximately
 $P_0=211.050\times10^6$. 
 The initial date of the time series of the data sets
 used is 03/20, one day after the first official death
 due to COVID-19 was recorded.
 We used the pre-pandemic birth and death rates shown
 in Table II. The estimates of contagion, lethality and
 recovery rates are obtained according to the methods described
 in Secs. \ref{conRate} and \ref{letRecRates}, respectively.
 The time variation of the contagion rate reflects the
 fact that the population slowly took heed of the gravity of the pandemic and started adopting social
 distancing measures and  using PPE.

 In Fig.~ \ref{casosConfirmadosBra}{\bf A}, we compare the official data (blue dots) of confirmed cases with
 the corresponding time series obtained from our proposed
 model. 
In frame {\bf B}, we show a comparison between the number of recovered cases (blue line) and the predicted number of recovered cases predicted by the
theoretical model.
The discrepancies in the fitting may have to do with delays in the confirmation of the recovered cases,
as we can see in the jumps that occurred from 06/07 to 06/08 
and from 07/01 to 07/02.
One possible source of systematic error, towards sub-notification of recovered cases, could occur in milder cases. 
Also, recovered outpatients may fail to take another test to confirm their recoveries.
In frame {\bf C}, we show a comparison between the number of death cases due to COVID-19 (in blue) and the number of deaths predicted by the theoretical model.
In frame {\bf D}, we plot the active cases
data along with the theoretical predictions based on our model.
The theoretical fit is not as good as in the previous
figures, but it is still quite reasonable.
We again  validate
our forecasting model with a 95\% confidence interval based on the Markov chain.
In all cases the epidemiological data fell within the prediction
range, but we had to backtrack 4 weeks in relation to the forecasting in Germany.
The problem seems to be the artificial lack of report of new recoveries that can be seen in frame ${\bf B}$.

In Fig.~\ref{kappa_t}{\bf A}, we show the time evolution of the contagion rate $\kappa(t)$.
The sharp drop of this rate is likely due to the increase
of isolation and social distancing that grew at the
second half of March in Brazil.
Since about 03/27 it has been decreasing slowly on average, 
despite the weekly modulations.
This is certainly due to better precautions by the population (isolation, social distancing, hand washing, and the increased use of PPEs). 
In frame {\bf B}, we plot the reproduction number as a function of time. It is basically a scaled version of $\kappa(t)$.
Since about the beginning of July, $R_0(t)$ has been modulating
around 1.
Although, this is not enough since the number of active cases is
still very high.
In frame {\bf C}, we show the time evolution of the lethality and recovery rates, $\lambda(t)$ and $\rho(t)=1/\tau-\lambda(t)$,
respectively.
One sees more fluctuations near the beginning of the time series
likely because there were less active cases then.
Also, one sees that the lethality rate is decreasing on average.
This might be related with the increased amount of
testing in Brazil, which is still far below the necessary though.
It could also be related with increased sub-notification of
death cases.
Another possibility is that the medical treatment and procedures for the more severe cases of COVID-19 are being
better treated since May.
Whatever the case, this behavior should be further investigated.
%In Fig.~\ref{casosMortesBra}, one sees a lot of fluctuations in the data. 
%This might be due to possible delays in the verification
%of the causes of death.
%For example, this likely happens when a patient dies before receiving the result of a COVID-19 exam.
% The theoretical model indicates that if there is not
% an increase in social distancing, in 20 days we will
% be having 2000 to 2500 of COVID-19 related daily deaths due to the pandemic in Brazil.

\FloatBarrier
\begin{figure}[ht]
    \centerline{\includegraphics[width=\textwidth]{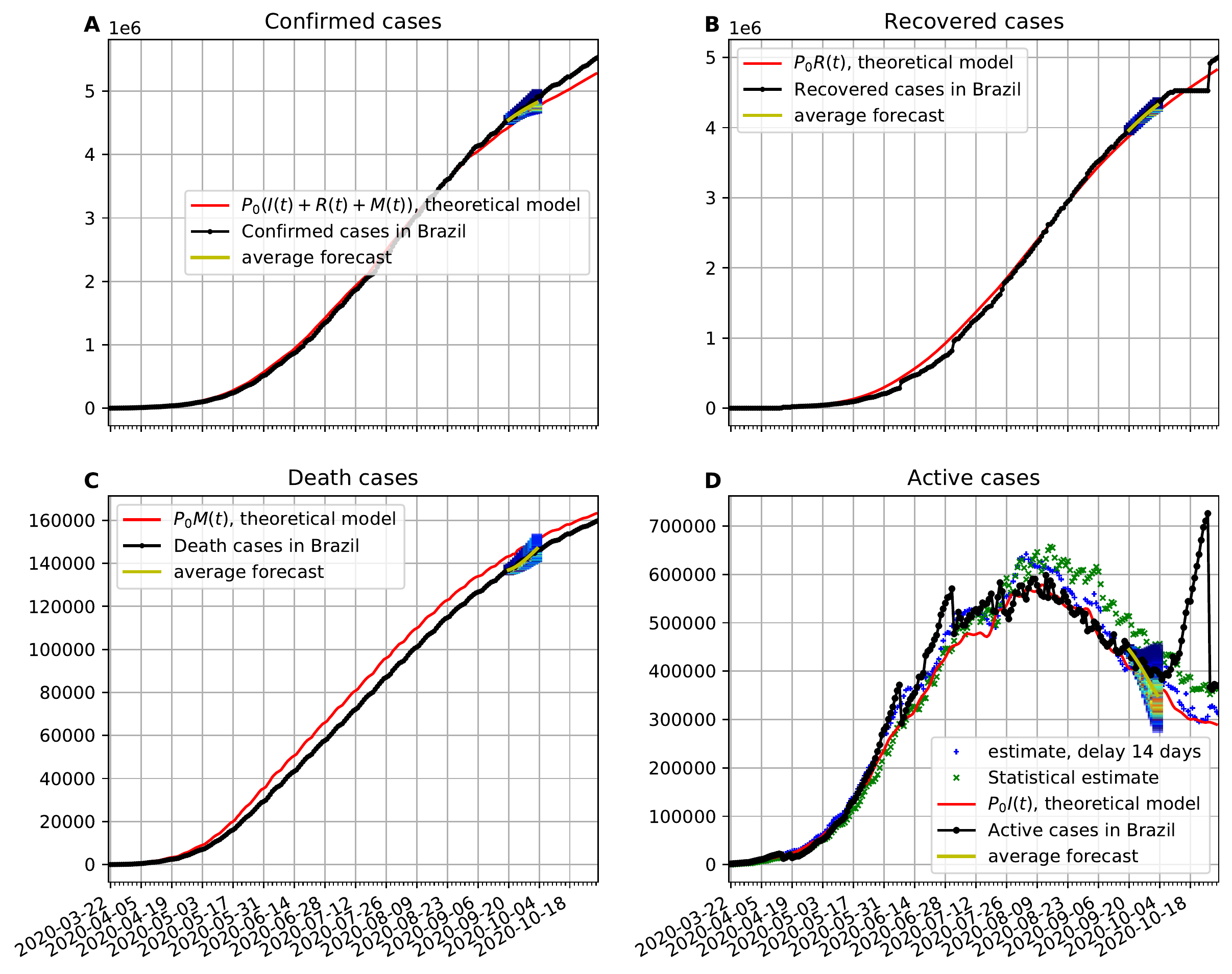}}
    %\centerline{\includegraphics[scale=0.8]{confirmed_Brazil.pdf}}
    \caption{ {\bf A} The number of official confirmed cases compared with the theoretical prediction. 
    {\bf B} Time evolution of the official number of recovered cases in comparison with the theoretical prediction.
    {\bf C}The number of official deaths compared with the 
    theoretical prediction.
    {\bf D} The time evolution of the official number of active cases in Brazil compared with the theoretical prediction.
    In all cases the functions $\kappa(t)$, $\lambda(t)$, and $\rho(t)$ vary in time according to Fig.~\ref{kappa_t}.
    }
    \label{casosConfirmadosBra}
   \end{figure} 
% \begin{figure}[ht]
%     \centerline{\includegraphics[scale=0.8]{deaths_Brazil.pdf}}
%     \caption{
%     The theoretical fit was obtained with the same parameters of Fig.~\ref{casosConfirmadosBra} and with the same contagion function $\kappa(t)$ of
%   Fig.~\ref{kappa_t}.
%     The best fit was obtained with the lethality probability
%     $P_\lambda=0.0947$ and the average time of infection $\tau=14$~days.
%     }
%     \label{obitosBra}
% \end{figure}

% \begin{figure}[ht]
%     \centerline{\includegraphics[scale=0.8]{recovered_Brazil.pdf}}
%     \caption{ The theoretical fit is obtained with the same parameters and with the same function $\kappa(t)$ that was used in Fig.~\ref{casosConfirmadosBra}.
%     }
%     \label{casosRecuperadosBra}
%     \end{figure}
%   \begin{figure}[ht]
%     \centerline{\includegraphics[scale=0.8]{active_Brazil.pdf}}
%     \caption{ The theoretical fit is obtained with the same parameters and with the same function $\kappa(t)$ that was used in Fig.~\ref{casosConfirmadosBra}.
%     }
%     \label{ativosBra}
%     \end{figure} 

\begin{figure}[ht]
    \centerline{\includegraphics[scale=0.8]{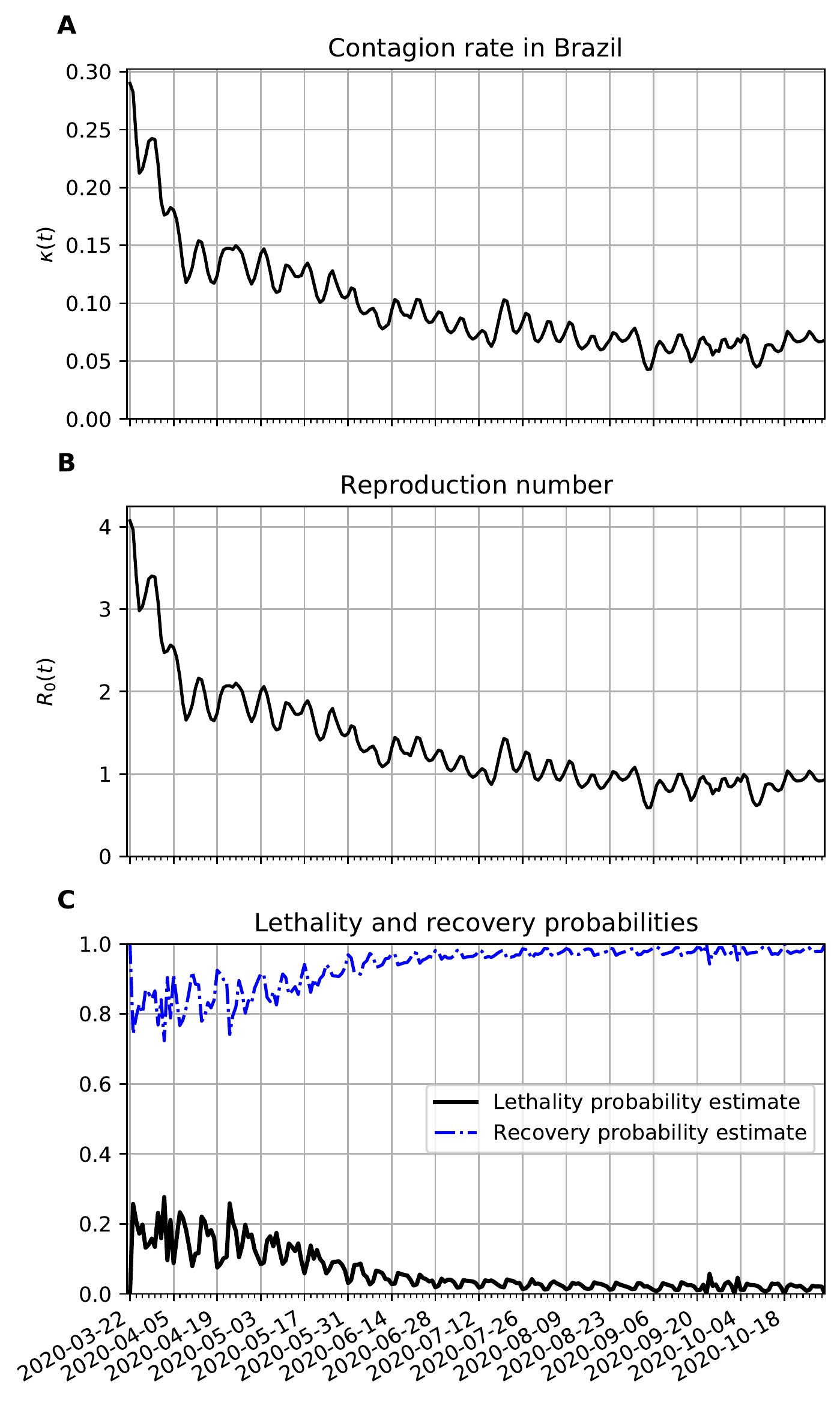}}
    \caption{
{\bf A} Time variation of the contagion rate  in Brazil.
    This function was obtained from the statistically estimated active cases as described in Sec. \ref{conRate}. The oscillations reflect weekly variations
    that can be seen in the number of daily new cases.
    {\bf B} Corresponding reproduction number evolution.
    {\bf C} The lethality and recovery rates obtained with the
    method described in Sec. \ref{letRecRates}.
    }
    \label{kappa_t}
    \end{figure}
\FloatBarrier

%\begin{figure}[ht]
%    \centerline{\includegraphics[scale=0.8]{{serieTemporalM}.pdf}}
%    \caption{
%    Deaths by day due the pandemic. The parameters and the contagion function are the same as was used in Fig.~\ref{casosConfirmadosBra}.
%    There are a lot of fluctuations in this data, this might be due delays in verification of the causes of
%    death.
%    }
%    \label{casosMortesBra}
%\end{figure}

\FloatBarrier

\subsection{Evolution of COVID-19 in Paraíba}
%We also provide this data in the supplementary material.
The initial population of Paraíba is $P_0=4,018,127$.
The first case of COVID-19 contamination was registered
on 03/18 and the first recorded death on 04/06.
We did not plot the number of recovered cases because we have not been able, so far, to obtain this data for Paraíba.

In Fig.~\ref{conRecDeaActPB}{\bf A}, we compare the official data (blue dots) of confirmed cases with
 the corresponding time series obtained from our
 model.
In frame {\bf B} , we show a comparison between the number of death cases due to COVID-19 (blue line) and the number of deaths predicted by the theoretical model.
In frame {\bf C}, we show the estimated active cases.
One result is based on delay and the other on the statistical
method.
Both methods are described in Sec.~\ref{stat_predict}.
We also show epidemiological model estimate (solid red line).
Note also, that at the last two weeks of the time series we validate
the forecasting model based on the Markov chain.
Here, we show a 95\% confidence band along two weeks.
In this case the all epidemiological data fell well within the predictions.

In Fig.~\ref{kappa_tPB}{\bf A}, we show the time evolution of the contagion rate $\kappa(t)$.
Initially, the contagion rate is not as high as in Brazil's case, but takes longer to drop.
From about 04/06, the rate of contagion starts decreasing on
average, although with a higher amplitude of modulation as in
the case of Brazil.
As time passes, one sees also a decrease in the amplitude of the
modulations.
This is certainly due to better precautions by the population (isolation, social distancing, hand washing, and the increased use of PPEs). 
In frame {\bf B}, we plot the reproduction number as a function of time. It is basically a scaled version of $\kappa(t)$.
Since about the beginning of July, $R_0(t)$ has been modulating
around 1.
Although, this is not enough to end the pandemic since the number of estimated active cases is still very high.
In frame {\bf C}, we show the time evolution of the lethality and recovery rates, $\lambda(t)$ and $\rho(t)=1/\tau-\lambda(t)$,
respectively.
Again, one sees more fluctuations near the beginning of the time series
likely because there were very few active cases then.
Also, one sees that the lethality rate started decreasing on average after the first week of April.
This might be related with the increased amount of
testing in Paraíba, which is still far below the necessary though.
\FloatBarrier
\begin{figure}[ht]
    \centerline{\includegraphics[width=0.9\textwidth]{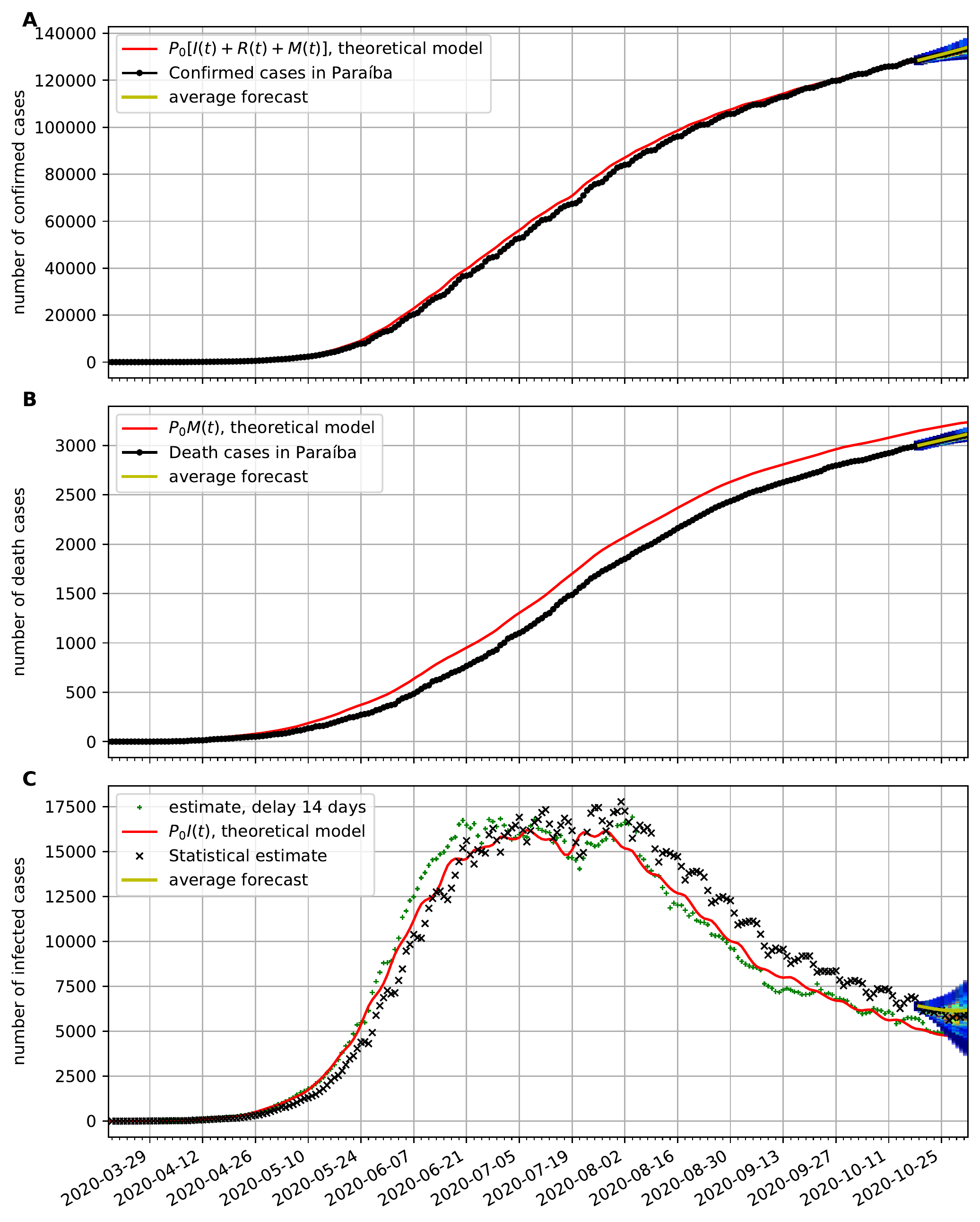}}
    %\centerline{\includegraphics[scale=0.8]{confirmed_PB.pdf}}
    %\centerline{\includegraphics[scale=0.8]{confirmadosAcumuladosPB.pdf}}
    \caption{
    The theoretical fit is obtained with $\tau=13$ and  the parameters  estimated in Fig.~\ref{kappa_tPB}.
    {\bf A} Comparison of the number of official confirmed cases with the theoretical prediction.
    {\bf B} Comparison of the number of official deaths due to COVID-19 infection with the theoretical prediction.
    {\bf C} Comparison of the number of the estimated number of active cases with the theoretical prediction.
    }
    \label{conRecDeaActPB}
\end{figure}

\begin{figure}[ht]
    \centerline{\includegraphics[scale=0.8]{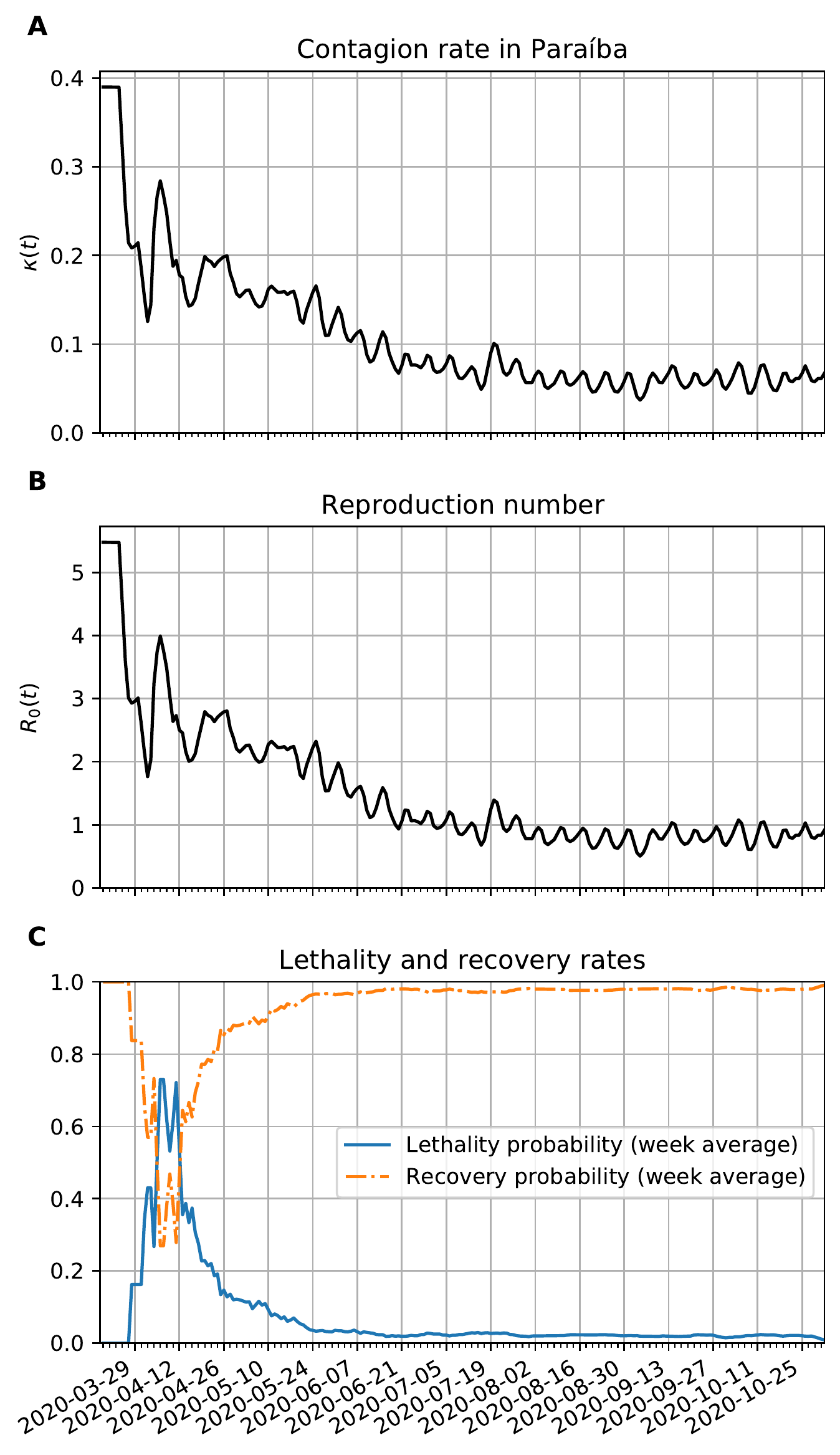}}
    %\centerline{\includegraphics[scale=0.8]{kappa_tPB.pdf}}
    \caption{{\bf A} Time variation of the contagion rate  in Paraíba.
    This function was obtained from the statistically estimated active cases described in Sec.~\ref{stat_predict} and the
    method of Sec.~\ref{conRate}.
    The initial cutoff value of $\kappa(t)$ is 0.38.
    {\bf B} Corresponding reproduction number evolution.
    {\bf C} Estimates for the lethality and recovery rates based on Sec.~\ref{letRecRates}.
    }    
    \label{kappa_tPB}
 \end{figure}   
\FloatBarrier

\FloatBarrier
\subsection{Evolution of COVID-19 in Campina Grande}
%We also provide this data in the supplementary material.
The first case of COVID-19 contamination was registered
on 03/27 and the first recorded death on 04/16.
We did not plot the number of recovered cases because we have not been able, so far, to obtain this data for Campina Grande.

In Fig.~\ref{conRecDeaActCG}{\bf A}, we compare the official data (blue dots) of confirmed cases with
 the corresponding time series obtained from our
 model.
In frame {\bf B} , we show a comparison between the number of death cases due to COVID-19 (blue line) and the number of deaths predicted by the theoretical model.
In frame {\bf C}, we show the estimated active cases.
One result is based on delay and the other on the statistical
method.
Both methods are described in Sec.~\ref{stat_predict}.
We also show the epidemiological model estimate (solid red line).
Note also, that at the last two weeks of the time series we again validate
the forecasting model based on the Markov chain.
Here, we show a 95\% confidence band along two weeks.
In this case the all epidemiological data fell well within the predictions.

In Fig.~\ref{kappa_tPB}{\bf A}, we show the time evolution of the contagion rate $\kappa(t)$.
Initially, the contagion rate is not as high as in Brazil's case, but takes longer to drop.
Only from about 06/07, the rate of contagion starts decreasing on
average, although with a higher amplitude of modulation as in
the case of Brazil or Paraíba.
As time passes, one sees also a decrease in the amplitude of the
modulations.
This is certainly due to better precautions by the population (isolation, social distancing, hand washing, and the increased use of PPEs). 
In frame {\bf B}, we plot the reproduction number as a function of time. It is basically a scaled version of $\kappa(t)$.
Since about the beginning of July, $R_0(t)$ has been modulating
around 1.
Although, this is not enough to end the pandemic since the number of estimated active cases is still very high.
In frame {\bf C}, we show the time evolution of the lethality and recovery rates, $\lambda(t)$ and $\rho(t)=1/\tau-\lambda(t)$,
respectively.
Again, one sees more fluctuations near the beginning of the time series
likely because there were very few active cases then.
Also, one sees that the lethality rate started decreasing on average after the first week of April.
This might be related with the increased amount of
testing in Campina Grande, which is still far below the necessary though.
\FloatBarrier
\begin{figure}[ht]
    \centerline{\includegraphics[width=0.9\textwidth]{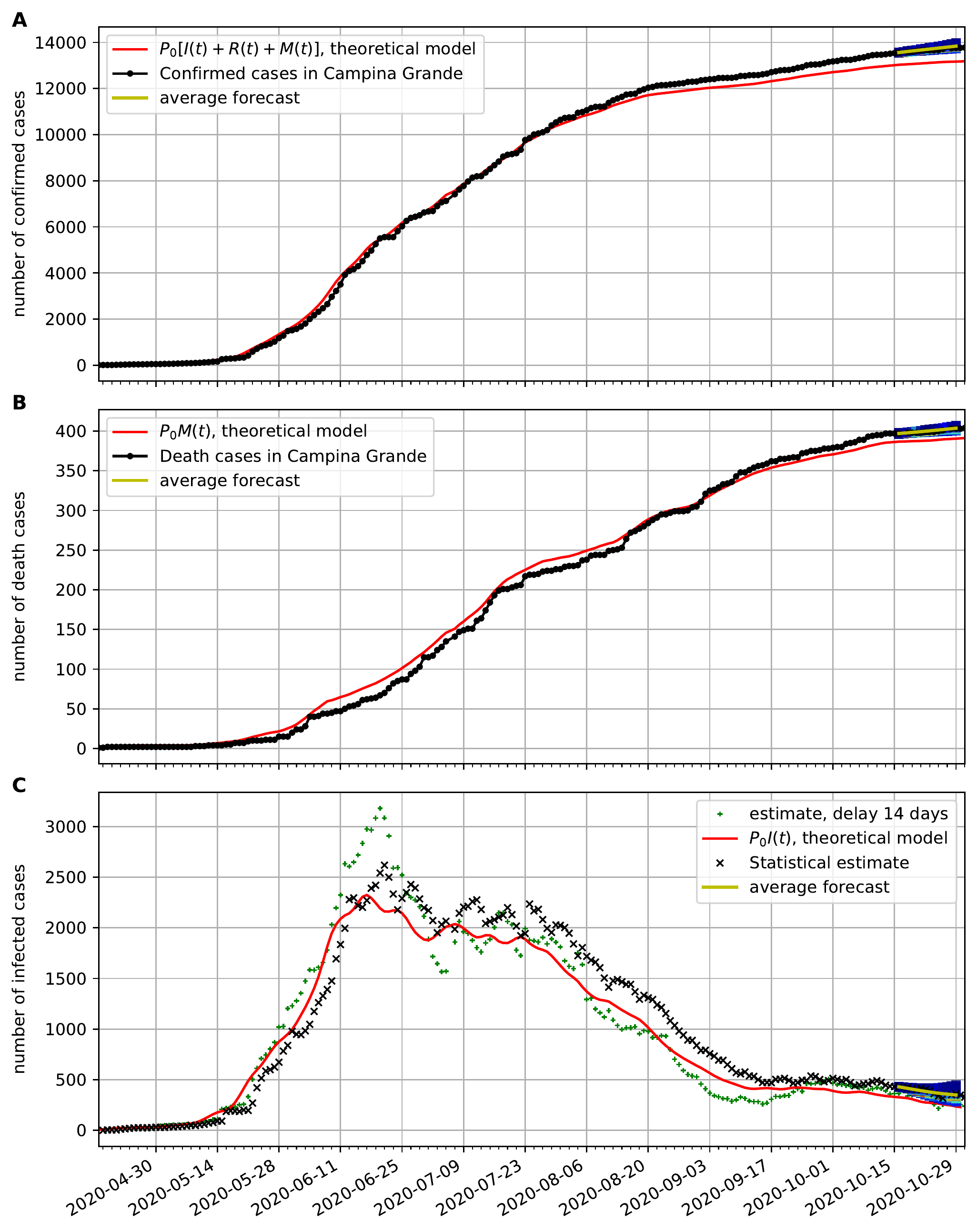}}
    \caption{
    The theoretical fit is obtained with $\tau=14$day and  the parameters  estimated in Fig.~\ref{kappa_tCG}.
    {\bf A} Comparison of the number of official confirmed cases with the theoretical prediction.
    {\bf B} Comparison of the number of official deaths due to COVID-19 infection with the theoretical prediction.
    {\bf C} Comparison of the number of the estimated number of active cases with the theoretical prediction.
    }
    \label{conRecDeaActCG}
\end{figure}

\begin{figure}[ht]
    \centerline{\includegraphics[scale=0.8]{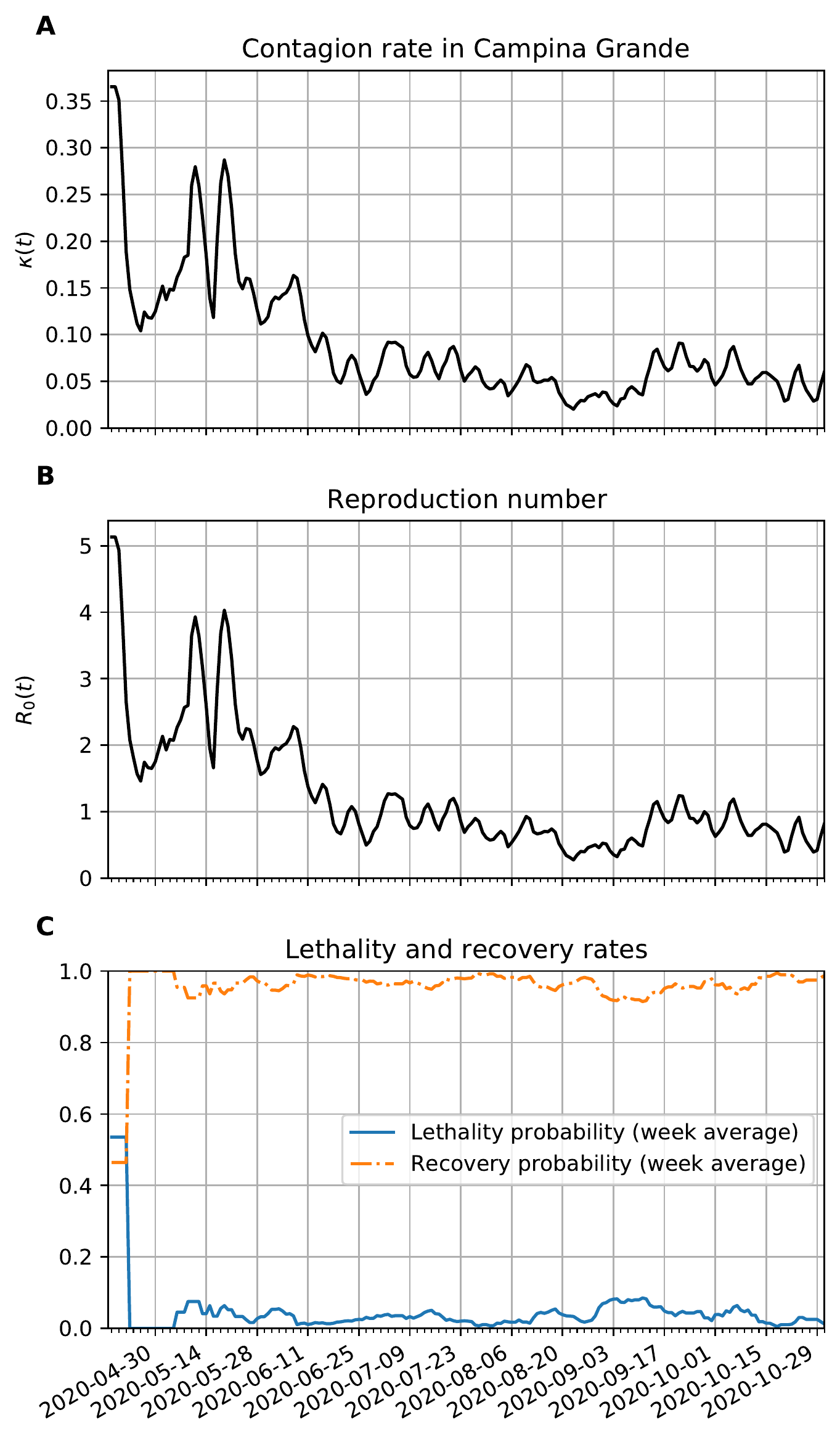}}
    \caption{{\bf A} Time variation of the contagion rate  in Paraíba.
    This function was obtained from the statistically estimated active cases described in Sec.~\ref{stat_predict} and the
    method of Sec.~\ref{conRate}.
    The initial cutoff value of $\kappa(t)$ is 0.3.
    {\bf B} Corresponding reproduction number evolution.
    {\bf C} Estimates for the lethality and recovery rates based on Sec.~\ref{letRecRates}.
    }    
    \label{kappa_tCG}
 \end{figure}   
\FloatBarrier
\section{Conclusion}
Here we summarize the main contributions of our epidemiological
study of the evolution of the COVID-19 pandemic in four populations. The proposed model is an adaptation of the SIR model  \cite{kermack1927contribution}, a SIRD model \cite{hethcote2000mathematics}, with some notable differences. 
In the SIRD model we replace the removed population by the recovered and the deceased.
We allow that the contagion rate varies in time so that it reflects the fact that
social distancing and isolation changes over time.
We developed two ways to obtain the active population from
the data on confirmed cases only.
In the first approach, the number of active cases at time $t$ is estimated simply by difference between the confirmed population at time $t$ minus the confirmed population at time $t-\tau$, where $\tau$ denotes the average time span of infection.
In the second approach, the number of active cases is estimated
by the probabilistic model proposed here.
Both approaches result in fairly accurate predictions of the active cases when the official data is available.
In cases in which the data on recovered cases is not available,
this approach could be the only way to estimate the number of
active cases.
We propose a new method to track automatically the contagion rate based on the statistical estimate of active cases.
We divided the data in weekly intervals and for each interval
we made a linear regression to obtain the slope, and based
on the equation in our model for the infected, we could obtain
a daily contagion rate. This contagion rate was fed back into
the equations of motion and we were able to fit the data
with a minimum of or no ad hoc interventions.
We noticed that the contagion rate could reach very high values
at the initial stages of the spread of the epidemic in all populations investigated.
This likely happens because of inherent numerical errors, as described in Sec.~\ref{conRate}, and multiple infections that are
imported from contaminated visitors or people returning from
trips abroad.
This type of dynamics is more relevant at the beginning, when
the local contaminated population is small and before barriers
on travelling are imposed by the governments.
Here, we also take into account the contribution from the pre-pandemic birth and death
rates to the evolution of the populations investigated. This could become relevant if the pandemic lasts for over a year and also
it is important as a comparison for the lethality of the pandemic.
According to the results exposed in the previous section, with our model we could
fit the official case data from Germany, Brazil, the State of Paraíba, and city of Campina Grande quite well.

The modeling of the spread of the pandemic in Germany is
very emblematic, since it clearly shows that the strict social distancing measures imposed by the government on 03/22/2020 were very effective in containing the disease,
reducing the $R_0$ from 2.8 down to roughly 1 about two weeks later, and further down to approximately 0.5 in more two weeks, according to our model.
 
 In the case of Brazil, we conclude that, based on the fit of the proposed model, the decrease of the reproduction number has been far more difficult and bumpier.
 This means the implanted social distancing measures are having an effect in thwarting the spread of the disease, but it has not enough, since to control the pandemic since $R_0(t)$
 has been oscillating around 1 and the number of active cases is
 still very high.
 
 The evolution of COVID-19 in the State of Paraíba has been similar to the national case. The contagion rate has been
 decreasing on average, but the modulations are even bumpier
 than in the national case.
 The amplitude of the modulations is considerably higher than
 in Brazil. This may be due to the smaller population involved.
 The effective reproduction number has been decreasing on average but it is still hovering around 1, with an estimated large number of active cases.

As was commented in the Introduction, Campina Grande adopted a social distancing policy one week earlier than the report of the first confirmed case.
Perhaps, due to this, the initial contagion rate was lower than
the corresponding ones in Brazil and Paraíba.
Despite of this, the rate of contagion did not decrease as fast as it did in Germany, Brazil or even in Paraíba.
It presents two major peaks in $\kappa(t)$ spaced apart by a week since the outbreak of the disease here.
Probably these could be linked to agglomeration events such
as in-branch governmental relief payments to unemployed people.
Also, one sees higher amplitude of the modulation of the contagion rate than in the other populations studied, this might
be linked to the smaller size of the population involved.

We also proposed and validated a simple forecasting method
 based on Markov chains and on our parameter estimation method for the evolution of the epidemiological
 data for up to two weeks. For the populations we investigated, the
 epidemiological data fell well within the 95\% confidence interval
 of our predictions.
 
Based on the results shown here, we conclude that the public health officials should
look into the local dynamics of the spread of the disease as they compare with the
theoretical predictions of models such as the one developed here.
In this way, they will know where the social distancing and isolation is being more
efficiently implemented. The models will be more relevant and accurate if there is more testing of the population. Even random testing should be considered, as one gains
statistical information on the spread of the disease and discovers where there is more
under-notification.
Also, cellphone data of the motion of the population, as
used by Peixoto \etal \cite{Peixoto2020} and Linka \etal \cite{Linka2020}, should be considered as a means
of predicting the contagion and of identifying hot spots of COVID-19.
This could help identify where there are more contagions: bars, churches, supermarkets, offices, pharmacies, hospitals, bakeries, restaurants, by delivery services, or family visits, etc. 
Furthermore, by comparing different local strategies one could gain insight on what
works better for slowing the spread of the disease.
It would be interesting to see the amount of correlation in
population mobility and contagion rate, this could be specially
relevant at the city level.
In a more refined work, one could couple several nearby cities into a network of populations.
%We hope this model can be adapted and applied mostly at the local level. It could help pinpoint hot spots in the contagion of COVID-19 and better identify which strategies work best in containing the spread of the disease.
%\section{Acknowledgments}
%The authors thank professors Francisco A. Brito, Antônio A. Lisboa de Souza, Michelli Barros, Joelson Campos, and Aldo Trajano for suggestions during the development of this work.
%\bibliography{biblio}%>> bibliography data
%\bibliographystyle{apsrev4-1} 

\end{document}